\documentclass[twoside,11pt]{article}

\usepackage{url}
\usepackage{amsmath}
\usepackage{paralist}
\usepackage{comment}
\usepackage{tikz}
\usetikzlibrary{matrix,arrows,decorations.pathmorphing}
\usepackage{subfig}
\let\chapter\section
\usepackage[ruled]{algorithm2e}
\usepackage{authblk}
\usepackage{natbib}
\usepackage{epsfig}
\usepackage{amssymb}
\usepackage[margin=1.3in]{geometry}
\makeatletter
\renewcommand*\env@matrix[1][c]{\hskip -\arraycolsep
    \let\@ifnextchar\new@ifnextchar
    \array{*\c@MaxMatrixCols #1}}
\makeatother
\newcommand{\BlackBox}{\rule{1.5ex}{1.5ex}}  
\newenvironment{proof}{\par\noindent{\bf Proof\ }}{\hfill\BlackBox\\[2mm]}
 
\newtheorem{theorem}{Theorem}
\newtheorem{lemma}[theorem]{Lemma}

\newtheorem{definition}[theorem]{Definition}

\def\comp{A^\star} 
\def\obs{A} 
\def\var{X} 
\def\opt{X^\star} 
\def\proj{\mathcal{P}}
\def\real{{\mathbb{R}}}
\def\loss{\ell\textit{oss }}
\def\dist{\mathcal{D}}
\def\uni{\dist_{\textrm{uni}}}
\def\power{\dist_{\textrm{pow}}}
\def\bb{{\bf b}}
\def\be{{\bf e}}

\def\bu{{\bf u}}
\def\bw{{\bf w}}
\def\bh{{\bf h}}
\definecolor{darkgray}{gray}{0.2}

\newcommand\sign{\operatorname{sign}} 
\newcommand\ind[1]{\mathbf{1}\left[ #1 \right]} 

\newcommand\cycle{\sigma} 
\newcommand\scycle{\tilde{\sigma}} 

\newcommand\cycleset[2]{\ensuremath{C_{#1}{\left( #2 \right)}}}  
\newcommand\scycleset[2]{\ensuremath{SC_{#1}{\left( #2 \right)}}} 

\renewcommand\path{\pi} 
\newcommand\pathset[2]{\ensuremath{\mathcal{P}_{#1}\left( #2 \right)}} 

\newcommand\moi[1]{\mu\left( #1 \right)} 
\newcommand\trianglemoi[1]{\ensuremath{\mu_{\mathrm{tri}}{\left( #1 \right)}}} 
\newcommand\orderkmoisimple[2]{\ensuremath{\mu^s_{#1}{\left( #2 \right)}}} 
\newcommand\orderkmoi[2]{\ensuremath{\mu_{#1}{\left( #2 \right)}}} 

\newcommand\Gplus[2]{\ensuremath{\obs^{+(#1,#2)}}} 
\newcommand\Gminus[2]{\ensuremath{\obs^{-(#1,#2)}}} 

\begin{document}

\title{Prediction and Clustering in Signed Networks:\\
A Local to Global Perspective}

\author[1]{Kai-Yang Chiang}
\author[1]{Cho-Jui Hsieh}
\author[1]{Nagarajan Natarajan}
\author[2]{Ambuj Tewari}
\author[1]{Inderjit S. Dhillon}
\affil[1]{Department of Computer Science, University of Texas, Austin}
\affil[2]{Department of Statistics, University of Michigan, Ann Arbor}

\renewcommand\Authands{ and }
\date{}

\maketitle

\begin{abstract}
The study of social networks is a burgeoning research area. However, most existing work deals with networks 
that simply encode whether
relationships exist or not. In contrast, relationships in \emph{signed} networks can be positive (``like", ``trust")
or negative (``dislike", ``distrust"). The theory of social balance shows that 
signed networks tend to conform to some local patterns that, in turn, induce
certain global characteristics. In this paper, we exploit both local as well as global aspects of social balance theory for two fundamental problems in the
analysis of signed networks: sign prediction and clustering.
Motivated by local patterns of social balance,
we first propose two families of sign prediction methods: measures of social
imbalance (MOIs), and supervised learning using high order cycles (HOCs). These methods predict signs
of edges based on triangles and $\ell$-cycles for relatively small values of $\ell$. Interestingly, 
by examining measures of social imbalance, we show that
the classic Katz measure, which is used widely in unsigned link prediction, actually has a balance theoretic interpretation when applied to
signed networks. Furthermore, motivated by the global
structure of balanced networks, we propose an effective low rank modeling approach for both 
sign prediction and clustering.
For the low rank modeling approach, we provide theoretical performance guarantees via convex relaxations,
scale it up to large problem sizes using a matrix factorization based algorithm, and provide extensive experimental validation including
comparisons with local approaches.
Our experimental results indicate
that, by adopting a more global viewpoint of balance structure, we get significant performance and computational gains in prediction and clustering
tasks on signed networks. Our work therefore highlights the usefulness of the global aspect of balance theory for the analysis of signed networks.

\end{abstract}

\section{Introduction}
\label{sec:introduction}
The study of networks is a highly interdisciplinary field that draws ideas and inspiration from
multiple disciplines including biology, computer science, economics, mathematics, physics, sociology, and statistics. 
In particular, \emph{social network analysis} deals with networks that form between people. With roots in sociology, social
network analysis has evolved considerably. Recently, a major force in its evolution has been the growing importance of online social
networks that were themselves enabled by the Internet and the World Wide Web. A natural result of the proliferation of online social networks has been the increased involvement
in social network analysis of people from computer science, data mining, information studies, and machine learning.

Traditionally, online social networks have been represented as graphs, with
nodes representing entities, and edges representing relationships between entities.  
However, when a network has like/dislike, love/hate, respect/disrespect, or trust/distrust
relatiships, such a representation is inadequate since it fails to encode the \emph{sign}
of a relationship.
Recently, online networks where two opposite kinds of relationships can occur have become common.
For example, online review websites such as Epinions allow users to either like or dislike
other people's reviews. Such networks can be modeled as {\it signed networks}, where
edge weights can be either greater or less than $0$, representing positive or negative relationships 
respectively. The development of theory and algorithms for signed networks is an important research
task that cannot be succesfully carried out by merely extending the theory and algorithms for unsigned networks
in a straightforward way. First, many notions and algorithms for unsigned networks break down
when edge weights are allowed to be negative. 
Second, there are some interesting theories that are applicable only to signed networks.

Perhaps the most basic theory that is applicable to signed social networks but does not appear in the study
of unsigned networks is that of {\it social balance} 
\citep{Harary53a, Cartwright56a}. The theory of social balance states that relationships in friend-enemy networks tend to
follow patterns such as ``an enemy of my friend is my enemy" and ``an enemy of my enemy is my friend".
A notion called weak balance \citep{Davis67a} further generalizes social balance
by arguing that in many cases an enemy of one's enemy can indeed act as an enemy.
Both balance and weak balance are defined in terms of {\it local} structure at the level of triangles.
Interestingly, the local structure dictated by balance theory also leads to a special 
{\it global} structure of signed networks.  We review the connection between local and global structure of balance signed
networks in Section \ref{sec:preliminaries}.

Social balance has been shown to be useful for prediction and clustering tasks for signed networks.
For instance, consider the {\it sign prediction problem} where the task is to predict the (unknown) sign
of the relationship between two given entities. Ideas derived from local balance
of signed networks can be succesfully used to yield algorithms for sign prediction~\citep{Leskovec10a, Chiang11a}.
In addition, the {\it clustering problem} of partitioning the nodes of a graph into tightly knit clusters
turns out to be intimately related to weak balance theory. We will see how a clustering into mutually antagonistic groups naturally
emerges from weak balance theory (see Theorem \ref{thm:weak_balance_theory} for more details).

The goal of this paper is to develop algorithms for prediction and clustering in signed networks
by adopting the local to global perspective that is already present in the theory of social balance.
What we find particularly exciting is that the local-global interplay that occurs in the {\it theory}
of social balance also occurs in our {\it algorithms}. We hope to convince the reader that, even though the local and global
definitions of social balance are theoretically equivalent, algorithmic and performance gains occur when a more global approach in algorithm design is adopted. 

We mentioned above that a key challenge in designing algorithms for signed networks is that
the existing algorithms for unsigned networks may not be easily adapted to the signed case.
For example, it has been shown that spectral clustering algorithms
for unsigned networks cannot, in general, be directly extended to signed networks~\citep{Chiang12a}.  
However, we do discover interesting connections between methods meant for unsigned networks and those meant for signed networks.
For instance, in the context of sign prediction, we see that
that the \emph{Katz measure}, which is widely used for unsigned link prediction, actually has a justification
as a sign prediction method in terms of balance theory. Similarly, methods based on \emph{low rank matrix completion}
can be motivated straight out of the global viewpoint of balance theory. Thus, we see that existing methods
for unsigned network analysis can reappear in signed network analysis albeit due to different reasons.

Here are the key contributions we make in this paper:
\begin{itemize}
  \item We provide a local to global perspective of the sign prediction problem,
  and demonstrate that our global methods are superior on synthetic as well as real-world data sets.
  \item In particular, we propose three sign prediction methods based on (i) measures of social imbalance (MOIs), 
  (ii) supervised learning using higher-order cycles (HOCs), and (iii) low-rank modeling. 
  The methods using higher-order cycles are more global than existing methods that just use triangles, 
  while the low-rank modeling approach can be viewed as a fully global approach motivated by global implications of structural balance.
  \item We show that the Katz measure used for unsigned networks can be interpreted from a social balance perspective: this immediately yields a sign prediction method.
  \item We provide theoretical guarantees for sign prediction and signed network clustering of balanced signed networks,
  under mild conditions on their structure.
\end{itemize}
Readers can find the preliminary versions of this paper in \citep{Chiang11a} and \citep{Hsieh12a}.  
The sign prediction methods based on paths and
cycles were first proposed in \citep{Chiang11a}, and low-rank modeling in \citep{Hsieh12a}.
In this paper, we provide a more detailed and unifying treatment of our previous research; in particular, we provide a
local-to-global perspective of the proposed methods, and a more comprehensive theoretical and experimental treatment.

The organization of this paper is guided by the local versus global aspects of social balance theory.
We first review some basics of signed networks and balance theory in Section \ref{sec:preliminaries}.
We recall notions such as (strong) balance and weak balance while emphasizing the connections between
local and global structures of balanced signed networks.  We will see that local balance structure is revealed by
triads (triangles) and cycles, while global balance structure manifests itself as clusterability of the nodes in the network. Based on these
observations, in Section \ref{sec:triads}, we start by showing how to use triads for
sign prediction. In Section \ref{sec:longer_cycles}, we go beyond
triangles and explore prediction methods based on cycles of length up to $\ell\ge 3$.  Under this broader
view, we can exploit information that is less localized around an edge whose sign we have to predict. The hope is that going global should give us
higher predictive accuracy.
We propose two classes of methods: those based on measures of social imbalance (MOIs) and
those that use supervised learning techniques to exploit existence of balance at the level of high order cycles (HOCs).
In Section \ref{sec:low_rank_model}, we present a completely non-local approach based on the global structure of balanced signed networks.
We show that such networks have 
low rank adjacency matrices, so that we can solve the sign prediction problem by reducing it to a low rank matrix completion problem.  Furthermore,
the low rank modeling approach can also be used for the clustering of signed networks.  In Section
\ref{sec:experiment}, we conduct several experiments, which show that global methods
(based on low rank models) generally have better performance than local methods (based on triads and cycles).
Finally, we discuss related work in Section \ref{sec:related}, and state our conclusions
in Section \ref{sec:conclusion}.

\section{Signed Networks and Social Balance}
\label{sec:preliminaries}
In this section, we set up our notation for signed networks, review the basic notions of balance theory,
and describe the two main tasks (sign prediction and clustering) that we address in this paper.

\subsection{Categories of Signed Networks}

The most basic kind of a signed network is a \emph{homogeneous} signed network.
Formally, a homogeneous signed network is represented as a graph with
the adjacency matrix $\obs \in \{-1,0,1\}^{n \times n}$, which denotes
relationships between entities as follows:
\begin{align*}
  \obs_{ij} = \begin{cases}
  1, &\text{if $i$ \& $j$ have positive relationship}, \\
  -1, &\text{if $i$ \& $j$ have negative relationship}, \\
  0, &\text{if relationship between $i$ \& $j$ is unknown (or missing)}.
  \end{cases}
\end{align*}
We should note that we treat a zero entry in $\obs$ as an \emph{unknown} 
relationship instead of no relationship, since we expect any two entities
have some (hidden) positive or negative attitude toward each other
even if the relationship itself might not be observed.  From an alternative
point of view, we can assume there exists an underlying {\em complete} signed network
$\comp$, which contains relationship information between all pairs of entities.  
However, we can only observe some partial entries of $\comp$, denoted by $\Omega$.
Thus, the partially observed network $\obs$ can be represented as:
\begin{align*}
  \obs_{ij} = \begin{cases}
  \comp_{ij}, &\text{if $(i,j) \in \Omega$}, \\
  0, &\text{otherwise}.
  \end{cases}
\end{align*}

A signed network can also be \emph{heterogeneous}.  In a heterogeneous signed network,
there can be more than one kind of entity, and relationships between
two, same or different, entities can be positive and negative.
For example, in the online video sharing website Youtube, 
there are two kinds of entities -- users and videos, and
every user can either \emph{like} or \emph{dislike} a video.  Therefore, the
Youtube network can be seen as a bipartite signed network, in which
all positive and negative links are between users and videos.

In this paper, we will focus our attention on homogeneous signed networks,
i.e. networks where relationships are between the same kind of entities.
For heterogeneous signed networks, it is possible to do some preprocessing
to reduce them to homogeneous networks.  For instance, in a Youtube network,
we could possibly infer the relationships between users based on their
taste of videos.  These preprocessing tasks, however, are not trivial.

In the remaining part of the paper, we will use the term ``network"
as an abbreviation for ``signed network", unless we explicitly specify otherwise.  
In addition, we will now mainly focus on undirected signed graphs
(i.e. $\obs$ is symmetric) unless we specify otherwise.  For a directed signed network,
a simple but sub-optimal way to apply our methods is by considering the symmetric 
network, $\sign(\obs+\obs^T)$.

\subsection{Social Balance}
A key idea behind many methods that estimate a high dimensional complex object from limited data is
the exploitation of \emph{structure}. In the case of signed networks, researchers have identified various kinds of
non-trivial structure \citep{Harary53a, Davis67a}.
In particular, one influential theory, known as social balance theory, states that relationships between 
entities tend to be {\it balanced}.
Formally, we say a triad (or a triangle) is {\it balanced} if it contains an even number of negative edges. This is in agreement
with beliefs such as ``a friend of my friend is more likely to be my friend" and ``an enemy of my friend is more
likely to be my enemy". The configurations of balanced and unbalanced triads are shown in Table \ref{tab:balanced_triad}.

\begin{table}
\centering
\begin{tabular}{c|c}
\footnotesize{Balanced triads} & \footnotesize{Unbalanced triads} \\ \hline
\tikz[scale=.7]{\draw (-.5,0) node[left] {$a$} -- node[left] {$+$} (0,.866) node[above] {$b$} -- node[right] {$+$} (.5,0) node[right] {$c$} -- node[below] {$+$} (-.5,0);}
\hspace{0.5cm}
\tikz[scale=.7]{\draw (-.5,0) node[left] {$a$} -- node[left] {$-$} (0,.866) node[above] {$b$} -- node[right] {$-$} (.5,0) node[right] {$c$} -- node[below] {$+$} (-.5,0);}

\hspace{0.25cm}&\hspace{0.25cm}
\tikz[scale=.7]{\draw (-.5,0) node[left] {$a$} -- node[left] {$-$} (0,.866) node[above] {$b$} -- node[right] {$-$} (.5,0) node[right] {$c$} -- node[below] {$-$} (-.5,0);}
\hspace{0.5cm}
\tikz[scale=.7]{\draw (-.5,0) node[left] {$a$} -- node[left] {$+$} (0,.866) node[above] {$b$} -- node[right] {$+$} (.5,0) node[right] {$c$} -- node[below] {$-$} (-.5,0);}
\\
\end{tabular}
\caption{Configurations of balanced and unbalanced triads.}
\label{tab:balanced_triad}
\end{table}

Though social balance specifies the patterns of triads, 
one can generalize the balance notion to general $\ell$-cycles.       
An $\ell$-cycle is defined as a simple path from 
some node to itself with length equal to $\ell$.
The following definition extends social balance to general $\ell$-cycles:
\begin{definition}[Balanced $\ell$-cycles]
  \label{def:balanced_cycle}
  An $\ell$-cycle is balanced iff it contains an even number of negative edges.
\end{definition}
Table \ref{tab:balanced_cycle} shows some instances of balanced and unbalanced
cycles based on the above definition.
To define balance for general networks, we first define the notion of balance for \emph{complete} networks:

\begin{table}
\centering
\begin{tabular}{c|c}
\footnotesize{Balanced cycles} & \footnotesize{Unbalanced cycles} \\ \hline
\tikz[scale=1.2]{\draw (-.3,0) node[left] {$a$} -- 
  node[left] {$+$} (-.3,.6) node[left] {$b$} -- 
  node[above] {$+$} (.3,.6) node[right] {$c$} -- 
  node[right] {$+$} (.3,0) node[right] {$d$} -- 
  node[below] {$-$} (-.3,0);}
\hspace{0.25cm}
\tikz[scale=1.2]{\draw (-.2,0) node[left] {$a$} -- 
  node[left] {$+$} (-.5,.5) node[left] {$b$} -- 
  node[left] {$-$} (0,.9) node[above] {$c$} -- 
  node[right] {$+$} (.5,.5) node[right] {$d$} -- 
  node[right] {$-$} (.2,0) node[right] {$e$} -- 
  node[below] {$+$} (-.2,0);}
\hspace{0.2cm}&
\hspace{0.2cm}
\tikz[scale=1.2]{\draw (-.3,0) node[left] {$a$} -- 
  node[left] {$+$} (-.3,.6) node[left] {$b$} -- 
  node[above] {$+$} (.3,.6) node[right] {$c$} -- 
  node[right] {$+$} (.3,0) node[right] {$d$} -- 
  node[below] {$-$} (-.3,0);}
  \hspace{0.25cm}
\tikz[scale=1.2]{\draw (-.2,0) node[left] {$a$} -- 
  node[left] {$+$} (-.5,.5) node[left] {$b$} -- 
  node[left] {$+$} (0,.9) node[above] {$c$} -- 
  node[right] {$+$} (.5,.5) node[right] {$d$} -- 
  node[right] {$+$} (.2,0) node[right] {$e$} -- 
  node[below] {$-$} (-.2,0);}
\end{tabular}
\caption{Some instances of balanced and unbalanced cycles.}
\label{tab:balanced_cycle}
\end{table}

\begin{definition}[Balanced complete networks]
  \label{def:balance_complete_network}
  A complete network is balanced iff all triads in the network are balanced.
\end{definition}

Of course, most real networks are not complete.  In other words, 
we expect that there are always some missing entries in the adjacency matrix. That is,
there exist $i,j$ such that $\obs_{ij} = 0$.
To define balance for general networks, we adopt the perspective of a missing value estimation
problem as follows:
\begin{definition}[Balanced networks]
  A (possibly incomplete) network is balanced iff it is possible to assign $\pm1$ signs
  to all missing entries in the adjacency matrix, such that the
  resulting complete network is balanced.
\end{definition}

So far, the notion of balance is defined by specifying patterns of {\it local} structures
in networks (i.e. the patterns of triads).  The following result from balance theory
shows that balanced networks actually have a nice 
{\it global} structure.
\begin{theorem}[Balance theory, \citep{Cartwright56a}]
  \label{thm:balance_theory}
  A network is balanced iff either (i) all edges are positive, 
  or (ii) we can divide nodes into two groups (or clusters), such that all edges within clusters
  are positive and all edges between clusters are negative.
\end{theorem}

Now we can revisit the definition of balanced $\ell$-cycles 
(Definition \ref{def:balanced_cycle}).  Under that definition, we can actually
verify if a network is balanced or not by looking at all cycles in the network
due to the following well-known theorem:
\begin{theorem}
\label{thm:balance_cycle}
  A network is balanced iff all its $\ell$-cycles are balanced.
\end{theorem}
\begin{proof}
  First we prove the forward direction. 
  If we are given a balanced network, then we can
  divide the nodes into two antagonistic groups $X, Y$ as Theorem \ref{thm:balance_theory}
  shows (note that one of $X, Y$ could be empty).
  Without loss of generality, given any $\ell$-cycle, we can traverse this cycle from
  an arbitrary node $i \in X$, and we will switch the group when passing a negative edge.
  After $\ell$ steps we will stop at node $i \in X$ again; therefore, in these $\ell$
  steps we can only pass an even number of negative edges to ensure we stop at group $X$.
  Thus, any $\ell$-cycle in this balanced network is balanced.

  To prove the other direction, we give a procedure that partitions the network into two antagonistic
  groups (say $X$ and $Y$) if all $\ell$-cycles in the network are balanced.  Without loss of generality we can
  assume the network has only one connected component.  We first pick an arbitrary node $i$ and
  mark it in group $X$, and try to mark the other nodes by performing a depth first search (DFS) from $i$.  When we 
  traverse an edge $(u,v)$, we mark $v$ as belonging to the same group as $u$ if $(u,v)$ is positive,
  otherwise we mark $v$ as belonging to the opposite group as $u$.  Since all cycles in the network are
  balanced, a node marked as $X$ will not be marked $Y$ later on when traversing cycles.
  Therefore, after all nodes are marked, we find two groups $X, Y$ such that all edges within 
  $X$ or $Y$ are positive and all edges between $X$ and $Y$ are negative.  By Theorem 
  \ref{thm:balance_theory}, we conclude that this network is balanced.

\end{proof}


\subsection{Weak Balance}
One possible weakness of social balance theory is that the defined balance 
relationships might be too strict.  In particular, researchers have argued that
the degree of imbalance in the triad with two positive edges (the fourth triad in
Table \ref{tab:balanced_triad}) is much stronger than that in the triad with all negative
edges (the third triad in Table \ref{tab:balanced_triad}). Thus, we can say that
the first three triads in Table \ref{tab:balanced_triad} are \emph{weakly balanced}.
Based on this observation,
by also allowing triads with all negative edges,
a weaker version of balance notion can be defined \citep{Davis67a}.

As in the case of (strong) social balance, 
we start with a definition of weak balance
in a complete network:
\begin{definition}[Weakly balanced complete networks]
  \label{def:weak_balance_complete_network}
  A complete network is weakly balanced iff all triads in the network are weakly balanced.
\end{definition}
The definition for general incomplete networks can be obtained
by adopting the perspective of a missing value 
estimation problem:
\begin{definition}[Weakly balanced networks]
  \label{def:weak_balance_network}
  A (possibly incomplete) network is weakly balanced iff it is possible to obtain
  a weakly balanced complete network by filling the missing edges in its
  adjacency matrix.
\end{definition}

Though Definitions \ref{def:weak_balance_complete_network} and 
\ref{def:weak_balance_network} define weak balance in terms of patterns of local triads,
one can show that weakly balanced networks have a special global structure, analogous to 
Theorem \ref{thm:balance_theory}:
\begin{theorem}[Weak balance theory, \citet{Davis67a}]
  \label{thm:weak_balance_theory}
  A complete network is weakly balanced iff either (i) all of its edges are positive,
  or (ii) we can divide nodes into $k$ clusters, such that all the edges within
  clusters are positive and all the edges between clusters are negative.
\end{theorem}
Thus, we say a network is $k$-weakly balanced if its nodes can be divided into $k$
clusters as specified in Theorem \ref{thm:weak_balance_theory}.  Note that 
when $k = 2$, this theorem simply reduces to Theorem 
\ref{thm:balance_theory}.

\subsection{Key Problems in Signed Networks}

As in classical social network analysis, 
we are interested in what we can infer given a signed network topology.  
In particular, we will focus on two core problems --- sign prediction and clustering.

In the \emph{sign prediction} problem, we intend to infer the unknown relationship between a
pair of entities $i$ and $j$ based on partial observations of the entire network of relationships.
More specifically, if we assume that we are given a (usually incomplete) network $\obs$
sampled from some underlying (complete) network $\comp$,
then the sign prediction task is to recover the sign patterns of one or more edges in $\comp$.
%
This problem bears similarity to the {\it structural link prediction problem} 
in \emph{unsigned} networks
\citep{Nowell07a}.
Note that the {\it temporal link prediction problem} has also been studied in the context of
an unsigned network evolving in time. The input to the prediction algorithm then consists of a series of networks (snapshots)
instead of a single network. We do not consider such temporal problems in this paper.

\emph{Clustering} is another important problem in network analysis.  
Recall that according to weak balance theory (Theorem \ref{thm:weak_balance_theory}), we can find $k$ groups such that
they are mutually antagonistic in a weakly balanced network.
Motivated by this, the clustering task in a signed network is trying to identify
$k$ antagonistic groups in the network,
such that most entities within the same cluster are friends while most entities
belonging to different clusters are enemies. 
Notice that since this (weak) balance notion only applies to
signed networks, most traditional clustering algorithms for unsigned networks 
cannot be directly applied.


\section{Local Methods: Exploiting Triads}
\label{sec:triads}

Since the basic definition of structural balance is in terms of triangles,  a natural approach for designing
sign prediction algorithms proceeds by reasoning locally in terms of unbalanced triangles. We first define a \emph{measure
of imbalance} based on the number of unbalanced triangles in the graph,

\begin{equation}
\label{eq:trianglemoidef}
\trianglemoi{\obs} := \sum_{\scycle \in \scycleset{3}{\obs}} \ind{\scycle \text{ is unbalanced}} \ ,
\end{equation}
where $\scycleset{3}{\obs}$ refers to the set of triangles in the network $\obs$. In general, we use
$\scycleset{\ell}{\obs}$ to denote the set of all $\ell$-cycles in the network $\obs$.
A definition essentially similar to the one above appears in the recent work of van de Rijt \cite[p. 103]{Rijt11} who observes that the equivalence
\begin{equation*}
\trianglemoi{\obs} = 0 \text{ iff } \obs \text{ is balanced}
\end{equation*}
holds only for complete graphs.

The basic idea of using a measure of imbalance for predicting the sign of a given query link $\{i,j\}$, such that $i \neq j$ and $\obs_{i,j} = 0$ is as follows.
Given the observed graph $\obs$ and query $\{i,j\}$, $i \neq j$, we construct two graphs: $\Gplus{i}{j}$ and $\Gminus{i}{j}$.
These are obtained from $\obs$ by setting $\obs_{ij}$ to $+1$ and $-1$ respectively.
Formally, these two augmented graphs can be defined as:
\begin{align}
\label{def:augmentedgraphs}
  \Gplus{i}{j}_{uv} = \begin{cases}
    1, &\text{if } (u,v) = (i,j) \\
    A_{uv}, &\text{otherwise}.
  \end{cases} \ \ \ \ \ &
  \Gminus{i}{j}_{uv} = \begin{cases}
    -1, &\text{if } (u,v) = (i,j) \\
    A_{uv}, &\text{otherwise}.
  \end{cases} \\
    \nonumber
\end{align}

Given a measure of imbalance, denoted as $\moi{\cdot}$, the predicted sign of $\{i,j\}$ is simply:
\begin{equation}
\label{eq:signpredictionrule}
\sign\left(
	\moi{\Gminus{i}{j}} - \moi{\Gplus{i}{j}} 
\right) \ .
\end{equation}

Note that, to be able to do this quickly, we should use a $\moi{\cdot}$ for which the quantity \eqref{eq:signpredictionrule}
is efficiently computable. For $3$-cycles, this is particularly easy.
For a given graph $G$ and a test edge $\{i,j\}$, we are interested in computing the sign of:
\[
	\sum_{\cycle \in \cycleset{3}{\Gminus{i}{j}}} \ind{\cycle} \quad - \sum_{\cycle \in \cycleset{3}{\Gplus{i}{j}}} \ind{\cycle}
\]
where we abuse notation by using the shorthand $\ind{\sigma}$ for $\ind{\sigma \text{ is unbalanced}}$.
Somewhat surprisingly, this simply amounts to computing the $(i,j)$ entry in the matrix $A^2$ where $A$ is the (signed)
adjacency matrix of $G$. In fact, a more general result will be discussed below (see Lemma~\ref{thm:power}).

A method derived from a measure of imbalance relies
on social balance theory for link prediction in signed networks.  However, real
world networks may not conform to the prediction of social balance
theory or may do so only to a certain extent.  To deal with this situation, we
can use measures of imbalance to derive {\em features} that can then be fed to a
supervised machine learning algorithm along with the signs of the known edges
in the network.

Indeed this is the approach pioneered by~\cite{Leskovec10a}. Their feature construction can be described as follows.
 Fix an edge $e =
(i,j)$. Consider an arbitary common neighbor (in an undirected sense) $k$ of
$i$ and $j$.  The link between $i$ and $k$ can be in $4$ possible
configurations:
\begin{align*}
i &\stackrel{+}{\rightarrow} k &
i &\stackrel{+}{\leftarrow} k \\
i &\stackrel{-}{\rightarrow} k &
i &\stackrel{-}{\leftarrow} k \ .
\end{align*}
Similarly, there are 4 possible configurations for the link between $k$ and $j$.
Thus, we can get a total of $16$ features for the edge $e$ by considering the number
of common neighbors $k$ in each of the $4 \times 4 = 16$ configurations.

This corresponds to a supervised variant
of the $k$-cycle method for $k = 3$. 
Let $A^+$ and $A^-$ be the matrices of positive and negative edges
such that $A = A^+ + A^-$.
In terms of matrix powers, these sixteen features
are nothing but the $(i,j)$ entries in the sixteen matrices:
\begin{gather}
\notag
\obs^{b_1} \cdot \obs^{b_2}  \\ 
\notag
\obs^{b_1} \cdot \left(\obs^{b_2}\right)^T \\ 
\label{eq:2OrderFeatures}
\left(\obs^{b_1}\right)^T \cdot \obs^{b_2} \\
\notag
\left(\obs^{b_1}\right)^T \cdot \left(\obs^{b_2}\right)^T \ ,
\end{gather}
where $b_1,b_2 \in \{\pm\}$, and $(A^{b_1})^T$ denotes the transpose of $A^{b_1}$.

Note that we have described the features of a {\em directed} edge $e = (i,j)$.
Social balance theory has mostly been concerned with undirected networks and hence methods based on measures
of imbalance can deal with undirected networks only. When we learn weights for features that are motivated by balance theory,
we are weakening our reliance on social balance theory but can therefore naturally deal with directed graphs.


\section{Going Global: Exploiting Longer Cycles}
\label{sec:longer_cycles}

For an incomplete graph, imbalance might manifest itself only if we look at longer simple cycles.
Accordingly, we define a higher-order analogue of \eqref{eq:trianglemoidef},
\begin{equation}
\label{eq:orderkmoisimpledef}
\orderkmoisimple{\ell}{\obs} := \sum_{i=3}^\ell \beta_i \sum_{\scycle \in \scycleset{i}{\obs}} \ind{\scycle \text{ is unbalanced}} \ ,
\end{equation}
where $\ell \ge 3$ and $\beta_i$'s are coefficients weighting the relative contributions of unbalanced simple cycles of different lengths.
If we choose a decaying choice of $\beta_i$, like $\beta_i = \beta^i$ for some $\beta \in (0,1)$, then we can even define
an infinite-order version,
\begin{equation}
\label{eq:orderinfmoisimpledef}
\orderkmoisimple{\infty}{\obs} := \sum_{i \ge 3} \beta_i \sum_{\scycle \in \scycleset{i}{\obs}} \ind{\scycle \text{ is unbalanced}} \ .
\end{equation}

It is clear that $\orderkmoi{\infty}{\cdot}$ is a genuine measure of imbalance in the sense formalized by the following theorem.
\begin{theorem}
Fix an observed graph $\obs$. Let $\beta_i > 0$ be any sequence such that the infinite sum in \eqref{eq:orderinfmoisimpledef} is well-defined.
Then, $\orderkmoi{\infty}{\obs} > 0$ iff $\obs$ is unbalanced.
\end{theorem}
\begin{proof}
This follows directly from Theorem~\ref{thm:balance_cycle}.
\end{proof}

This suggests that we could use $\orderkmoi{\infty}{\cdot}$ as a measure of imbalance to derive sign prediction algorithms. However,
enumerating simple cycles of a graph is NP-complete
\footnote{By straightforward reduction to Hamiltonian cycle problem \citep{Karp72a}.}.
To get around this computational issue, we
slightly change the definition of $\orderkmoi{\ell}{\cdot}$ to the following.
\begin{equation}
\label{eq:orderkmoidef}
\orderkmoi{\ell}{\obs} := \sum_{i=3}^\ell \beta_i \sum_{\cycle \in \cycleset{i}{\obs}} \ind{\cycle \text{ is unbalanced}} \ .
\end{equation}
As before, we allow $\ell=\infty$ provided the $\beta_i$'s decay sufficiently rapidly.
\begin{equation}
\label{eq:orderinfmoidef}
\orderkmoi{\infty}{\obs} := \sum_{i \ge 3} \beta_i \sum_{\cycle \in \cycleset{i}{\obs}} \ind{\cycle \text{ is unbalanced}} \ .
\end{equation}
The only difference between these definitions and \eqref{eq:orderkmoisimpledef},\eqref{eq:orderinfmoisimpledef}
is that here we sum over \textit{all} cycles (denoted by~\cycleset{i}{\obs}), not just simple ones.
However, we still get a valid notion of imbalance as stated by the following result.
\begin{theorem}
Fix an observed graph $\obs$. Let $\beta_i > 0$ be any sequence such that the infinite sum in \eqref{eq:orderinfmoidef} is well-defined.
Then, $\orderkmoi{\infty}{\obs} > 0$ iff $\obs$ is unbalanced.
\end{theorem}
\begin{proof}
One direction is trivial. If $\obs$ is unbalanced then there is an unbalanced simple cycle. However, any simple cycle is obviously a cycle and hence the
sum in \eqref{eq:orderinfmoidef} will be strictly positive.

For the other direction, suppose $\orderkmoi{\infty}{\obs} > 0$. This implies there is an unbalanced cycle $\cycle$ in the graph. Decompose the unbalanced cycle into
finitely many simple cycles. We will be done if we could show that one of these simple cycles has to be unbalanced.
It is easy to see why this is true: if all of these simple cycles were balanced, they all would have had an even number of negative edges,
but then the total number of negative edges in $\cycle$ could not have been odd.
\end{proof}

\subsection{Katz Measure Works for Signed Networks}

The classic method of \citet{Katz53a} has been used successfully for \emph{unsigned} link prediction \citep{Nowell07a}. However, by considering a sign prediction method
based on $\orderkmoi{\infty}{\cdot}$ we obtain an interesting interpretation of the Katz measure on a signed network from a balance theory viewpoint. The following result
is the key to such an interpretation.

\begin{lemma}
\label{thm:power}
Fix $\obs$ and let $i\neq j$ be such that $(i,j) \notin \Omega$. Let $\Gplus{i}{j}$
and $\Gminus{i}{j}$ be the augmented graphs as defined in \eqref{def:augmentedgraphs}. Then, for any $\ell \ge 3$,
\[
	 \sum_{\cycle \in \cycleset{\ell}{\Gminus{i}{j}}} \ind{\cycle} \quad
	- \sum_{\cycle \in \cycleset{\ell}{\Gplus{i}{j}}} \ind{\cycle}
	= \obs^{\ell-1}_{i,j} \ .
\]
\end{lemma}
\begin{proof}
Define the sets of $\ell$-cycles,
\begin{align*}
C_\ell^+(i,j) &:= \{ \cycle \in \cycleset{\ell}{\Gplus{i}{j}} : \cycle \text{ includes } (i,j) \} \\
C_\ell^-(i,j) &:= \{ \cycle \in \cycleset{\ell}{\Gminus{i}{j}} : \cycle \text{ includes } (i,j) \} \ ,
\end{align*}
that include the edge $(i,j)$. Note that, since $\Gplus{i}{j}$ and $\Gminus{i}{j}$ only differ
in the sign of the edge $(i,j)$, we have,
\[
	\cycleset{\ell}{\Gplus{i}{j}} \backslash C_\ell^+(i,j) = 
	\cycleset{\ell}{\Gminus{i}{j}} \backslash C_\ell^-(i,j) \ .
\] 
Thus, we have,
\begin{align}
\notag
&\quad  \sum_{\cycle \in \cycleset{\ell}{\Gminus{i}{j}}} \ind{\cycle}
	- \sum_{\cycle \in \cycleset{\ell}{\Gplus{i}{j}}} \ind{\cycle} \\
\notag
&=      \sum_{\cycle \in C_\ell^-(i,j)} \ind{\cycle} + \sum_{\cycle \in \cycleset{\ell}{\Gminus{i}{j}} \backslash C_\ell^-(i,j)} \ind{\cycle} 
-\sum_{\cycle \in C_\ell^+(i,j)} \ind{\cycle} - \sum_{\cycle \in \cycleset{\ell}{\Gplus{i}{j}} \backslash C_\ell^+(i,j)} \ind{\cycle} \\
\label{eq:beforepaths}
&=      \sum_{\cycle \in C_\ell^-(i,j)} \ind{\cycle} - \sum_{\cycle \in C_\ell^+(i,j)} \ind{\cycle} \ .
\end{align}
Now cycles in $C_\ell^-(i,j)$ are in $1$-$1$ correspondence with paths $\path$ in
$\pathset{\ell-1}{i,j}$ of length $\ell-1$, in the original graph $\obs$, that go from
$i$ to $j$.  Moreover, $\cycle \in C_\ell^-(i,j)$ is unbalanced iff the
corresponding path in $\pathset{\ell-1}{i,j}$ has an {\em even} number of $-1$'s.
Similarly, $\cycle \in C_\ell^+(i,j)$ is unbalanced iff the corresponding path in
$\pathset{\ell-1}{i,j}$ has an {\em odd} number of $-1$'s. Thus, continuing
from~\eqref{eq:beforepaths}:
\begin{align*}
&\quad  \sum_{\cycle \in C_\ell^-(i,j)} \ind{\cycle} - \sum_{\cycle \in C_\ell^+(i,j)} \ind{\cycle} \\
&=      \sum_{\path \in \pathset{\ell-1}{i,j}} \ind{\path \text{ has even no. of $-1$'s}} 
-\sum_{\path \in \pathset{\ell-1}{i,j}} \ind{\path \text{ has odd no. of $-1$'s}} \\
&=      \sum_{i_1,i_2,\ldots,i_{\ell-2}} A_{i,i_1} \cdot A_{i_1,i_2} \cdot \ldots \cdot A_{i_{\ell-2},j} \\
&=      \left( A^{\ell-1} \right)_{i,j} \ ,
\end{align*}
where the second equality is true because $\obs$ only has $\pm1,0$ entries.
\end{proof}

Using Lemma~\ref{thm:power}, it is easy to see that the prediction 
\eqref{eq:signpredictionrule} using \eqref{eq:orderkmoidef} reduces to
\begin{equation*}
\sign\left( \orderkmoi{\ell}{\Gminus{i}{j}} - \orderkmoi{\ell}{\Gplus{i}{j}}  \right) 
= \sign\left( \sum_{t=3}^\ell \beta_t \obs^{t-1}_{i,j}  \right) \ .
\end{equation*}

Similar to the $\ell$-cycle case, the prediction~\eqref{eq:signpredictionrule} using \eqref{eq:orderinfmoidef} reduces to 
\begin{equation}
\sign\left( \orderkmoi{\infty}{\Gminus{i}{j}} - \orderkmoi{\infty}{\Gplus{i}{j}}  \right) 
= \sign\left( \sum_{\ell\ge3} \beta_{\ell} \obs^{\ell-1}_{i,j}  \right) \
  \label{eq:reduceorderinfmoi}
\end{equation}
using Lemma~\ref{thm:power}.

Following the above reduction, the connection between Katz measure and $\orderkmoi{\infty}{\cdot}$ stands out.
This connection is stated as the following theorem:

\begin{theorem}[Balance Theory Interpretation of the Katz Measure] 
Consider the sign prediction rule \eqref{eq:signpredictionrule} using 
$\orderkmoi{\infty}{\cdot}$ in the reduced form \eqref{eq:reduceorderinfmoi}.
In the special case when $\beta_\ell = \beta^{\ell-1}$ with $\beta$ small enough ($\beta < 1/\|\obs\|_{2}$), 
the rule can be expressed as
the Katz prediction rule for edge sign prediction, in closed form:
\begin{equation*}
\sign\left( \left( (I - \beta \obs)^{-1} - I - \beta \obs \right)_{i,j} \right) \ .
\end{equation*}
\end{theorem}

The Katz prediction rule has been successfully used as a link prediction method for {\em unsigned} networks~\citep{Nowell07a} but here we see it reappearing
for link prediction in {\em signed} networks from a social balance point of view. We find this connection between
Katz measure and social balance intriguing. 

\subsection{Learning the Weights}

As noted in Section \ref{sec:triads}, \citet{Leskovec10a} used triangle-based features to learn weights using a supervised learning method.
A criticism against using only these triangle-based features is that there could
be many people in the social network who do not share friends. In fact, this is
the case in most of the networks that are used by \citet{Leskovec10a}. The reason
their method is able to predict well on such pairs is that they additionally use
seven other ``degree-type'' features like in-degree and out-degree (and their signed
variants). Thus, the prediction for an edge with zero emdeddedness (embeddedness refers
to the number of common neighbors of the vertices of an edge) relies completely
on the degree-based features.
These degree features could possibly introduce a bias in learning. For example, a node that is
predisposed to make positive relationships, biases the classifier to predict
positive relationships.

This criticism thus necessitates incorporating features from higher-order cycles.
Generalizing the construction \eqref{eq:2OrderFeatures}, we can define
$64$ fourth-order features (corresponding to $4$-cycles in the graph) of an edge $(i,j)$ as the $(i,j)$ entries in the
matrices:
\begin{equation}
\label{eq:3OrderFeatures}
\left( \obs^{b_1}\right)^{t_1} \cdot \left( \obs^{b_2}\right)^{t_2} \cdot \left( \obs^{b_3}\right)^{t_3} \ ,
\end{equation}
where $b_i \in \{\pm\}$ indicates whether we look at the positive or negative part of $\obs$
and $t_i \in \{T,1\}$ indicates whether or not we transpose it. There are $4$ possibilities for
each $b_i,t_i$ pair, resulting in a total of $4 \times 4 \times 4 = 64$ possibilities.

By now the reader can guess the construction of features of a general order $\ell \ge 3$.
For the edge $(i,j)$, they will be the $(i,j)$ entries in the $4^{\ell-1}$ matrices
\begin{equation}
\left( \obs^{b_1}\right)^{t_1} \cdot \left( \obs^{b_2}\right)^{t_2} \ldots \cdot \left( \obs^{b_{\ell-1}}\right)^{t_{\ell-1}} \ ,
\label{eq:kOrderFeatures}
\end{equation}
with $b_i \in \{\pm\}, t_i \in \{T,1\}$.

Note that the number of features is exponential in $\ell$, and therefore it is not feasible to obtain features from
arbitrarily long cycles. 
We use $\ell \leq 5$ for supervised HOC methods in our experiments
that are presented in Section~\ref{sec:experiment}.

\subsection{Reducing the Number of Features}
\label{sec:reduce_features}
The number of features can quickly become unmanageable, and computationally infeasible, as soon as $\ell$ is beyond $5$. While dimensionality of the feature space may be the primary concern, the combinatorial nature of the features also raises the following intuitive concern: the interpretability of features rendered by high-order cycles, say when $\ell = 6$, composed of different signs and directions, is a challenge. For example, it is intuitively hard to appreciate the difference between the two walks 
$i \stackrel{+}{\rightarrow} k_{1} \stackrel{+}{\rightarrow} k_{2} \stackrel{-}{\rightarrow} k_{3} \stackrel{+}{\rightarrow} k_{4}\stackrel{+}{\rightarrow} j$ and $i \stackrel{+}{\rightarrow} k_{1} \stackrel{+}{\rightarrow} k_{2} \stackrel{-}{\leftarrow} k_{3} \stackrel{+}{\rightarrow} k_{4}\stackrel{+}{\rightarrow} j$.

With this realization, one way to quickly reduce the number of features, yet retain the information in longer cycles, is to consider the underlying undirected graph, ignoring the directions. In particular, the $\ell^{\text{th}}$ order features will be from the matrices
\begin{equation}
\obs^{b_1} \cdot \obs^{b_2} \ldots \cdot \obs^{b_{k-1}},
\label{eq:kOrderFeaturesCollapsed}
\end{equation}
with $b_i \in \{\pm\}$.
Note that since we are considering the undirected graph, we ensure that the features are symmetric by summing features of the form $\obs^{b_{1}}\obs^{b_{2}}$ and $\obs^{b_{2}}\obs^{b_{1}}$. Thus the number of $\ell^{\text{th}}$ order features to compute is reduced to $O(2^{\ell})$ from $O(4^{\ell})$. Though the number of features is still exponential in $\ell$, the construction of features becomes easier for small values of $\ell$.

We note that another way to avoid dealing with too many features is to use a {\em kernel} instead. A kernel
computes inner products in feature space without explicitly constructing the feature map. One can then use
off-the-shelf SVM classifiers to perform the classification. We leave this very promising approach of directly defining a
kernel on {\em pairs of nodes} of a graph and using it for link prediction to future work.

\subsection{Classifier} 
We use a simple logistic regression where the imbalance of an edge is modeled as a linear combination of the features, which are imbalances in cycles of various lengths and characteristics themselves. Let $V$ be the set of vertices in the network and $\Phi: V \times V \to \mathbb{R}^p$ denote the feature map. 
Then, 
\[ P(A_{ij} = +1 ) = \frac{1}{1 + \exp{\left(-w_0 - \langle \bw, \Phi(i,j) \rangle \right)}}, \]
using which logistic regression is used to learn $w_0$ and the weight vector $\bw = [w_1 \cdots w_p]^T \in \mathbb{R}^{p}$.
The prediction of any query $(i,j)$ is then given by $\sign(P(A_{ij} = +1) - 0.5)$.


\section{Fully Global: Low Rank Modeling}
\label{sec:low_rank_model}
In Section \ref{sec:longer_cycles}, we have seen how to use $\ell$-cycles
for sign prediction.  We have also seen that $\ell$-cycles play a major role in how 
balance structure manifests itself \emph{locally}. By increasing $\ell$, the level at which balance structure is
considered becomes less localized. Still, it is natural to ask whether
we can design algorithms for signed networks by directly making use of 
their {\em global} structure. To be more specific, let us revisit 
the definition of complete weakly balanced networks 
(notice that balance is a special case of weak balance).  In general, 
complete weakly balanced networks can be defined from either a local or a global 
point of view.
From a local point of view, a given network is weakly balanced if all {\em triads} are weakly 
balanced, whereas from a global point of view, a network is weakly balanced
if its {\em global structure} obeys the clusterability property stated in Theorem \ref{thm:weak_balance_theory}.
Therefore, it is natural to ask whether we can directly use this global structure
for sign prediction.
In the sequel, we show that weakly balanced
networks have a ``low-rank" structure, so that 
the sign prediction problem can
be formulated as a low rank matrix completion problem.
  
We begin by showing that given a complete $k$-weakly balanced network, 
its adjacency matrix $\comp$ has rank at most $k$:

\begin{theorem}[Low Rank Structure of Signed Networks]
  \label{thm:low_rank}
  The adjacency matrix $\comp$ of a complete $k$-weakly balanced network has rank $1$ if $k \leq 2$, and has rank $k$ for all $k > 2$.
\end{theorem}

\begin{proof}
  Since $\comp$ is $k$-weakly balanced, the nodes can be divided into $k$ groups, say $S^{(1)}, S^{(2)},\dots,S^{(k)}$. Suppose group $S^{(i)}$ contains nodes $s^{(i)}_1, s^{(i)}_2,\dots,s^{(i)}_{n_i} $,  then the column vectors $\comp_{:,s^{(i)}_1},\dots,\comp_{:,s^{(i)}_{n_i}}$ all have the following form (after suitable reordering of nodes):
  \begin{equation*}
    \bb_i = [-1 \quad \cdots \quad -1 \underbrace{1 \quad \cdots \quad 1}_{\text{the $i^{th}$ 
    group}} -1 \quad \cdots \quad -1]^T,
  \end{equation*}
  and so the column space of $\comp$ is spanned by $\{\bb_1,\dots,\bb_k\}$. 

  First consider $k \leq 2$, i.e., the network is strongly balanced.  If $k = 1$, it is easy to see that $\text{rank}(\comp) = 1$.  If $k = 2$, then $\bb_1 = -\bb_2$.  Therefore, $\text{rank}(\comp)$ is again $1$.

  Now consider $k > 2$.  In this case, we argue that $\text{rank}(\comp)$ exactly equals $k$ 
  by showing that $\bb_1,\dots,\bb_k$ are linearly independent. We consider the following $k \times k$ square matrix: 
  \begin{equation*}
    M=
    \begin{bmatrix}[r]
      1 & -1 & \cdots & -1 & -1 \\
      -1 & 1 & \cdots & -1 & -1 \\
      \vdots & \vdots & \ddots & \vdots & \vdots \\
      -1 & -1 & \cdots & 1 & -1 \\
      -1 & -1 & \cdots & -1 & 1
    \end{bmatrix}.
  \end{equation*}
  It is obvious that ${\bf 1} = [1 \  1 \cdots 1]^T$ is an eigenvector of $M$ with eigenvalue 
  $-(k-2)$. 
  We can further construct the other $k-1$ linearly independent eigenvectors, each with eigenvalue~$2$: 
\begin{equation*}
  \be_1-\be_2, \ \be_1-\be_3, \ \dots, \ \be_1-\be_k,  
\end{equation*}
where $\be_i \in \real^{k}$ is the $i^{th}$ column of the $k \times k$ identity matrix. These $k-1$ eigenvectors are
clearly linearly independent.  
Therefore, $\text{rank}(M)=k$.

From the above we can show that $\text{rank}(\comp) = k$.  Suppose that $\bb_1,\dots,\bb_k$  are not linearly independent, then there exists $\alpha_1,\dots,\alpha_k$, with some $\alpha_i \neq 0$, such that $\sum_{i=1}^k \alpha_i \bb_i = 0$. 
Using this set of $\alpha$'s, it is easy to see that $\sum_{i=1}^k \alpha_i M_{:,i}=0$,  but this contradicts the fact that $\text{rank}(M)=k$.  Therefore, 
$\text{rank}(\comp)=k$. 
\end{proof}

\begin{figure}[th]
\centering
\tikzset{
  node style plus/.style={rectangle, left,font=\color{black}, minimum size=6pt},
       node style minus/.style={rectangle, left,font=\color{darkgray}, minimum size=6pt},
       arrow style mul/.style={draw,sloped,midway,fill=white},
       arrow style plus/.style={midway,sloped,fill=white},
}
\begin{tikzpicture}[>=latex]
{\footnotesize
\matrix (A) [matrix of math nodes,%
nodes = {node style minus},%
         left delimiter  = (,%
         right delimiter = )] at (0,0)
{%
  \node[node style plus]{1}; &
  \node[node style plus]{1}; &
  \node[node style plus]{1}; &
  -1 & -1 & -1 \\
  \node[node style plus]{1}; &
  \node[node style plus]{1}; &
  \node[node style plus]{1}; &
  -1 & -1 & -1 \\
  \node[node style plus]{1}; &
  \node[node style plus]{1}; &
  \node[node style plus]{1}; &
  -1 & -1 & -1 \\
  -1 & -1 & -1 & \node[node style plus]{1}; & -1 & -1 \\  
  -1 & -1 & -1 & -1 & 
  \node[node style plus]{1}; &
  \node[node style plus]{1}; \\
  -1 & -1 & -1 & -1 & 
  \node[node style plus]{1}; &
  \node[node style plus]{1}; \\
};

\node(equal)[right=10pt] at (A.east){=};

\matrix (B) [matrix of math nodes,right=10pt,
nodes = {node style minus},%
         left delimiter  = (,%
         right delimiter = )] at (equal.east)
{
  -1.045&0.265&-0.402\\
    -1.045&0.265&-0.402\\
    -1.045&0.265&-0.402\\
    0.319&-1.260&-0.830\\
    0.919&0.670&-0.541\\
    0.919&0.670&-0.541\\
};

\matrix (D) [matrix of math nodes, right = 18pt,
nodes = {node style minus},%
         left delimiter  = (,%
         right delimiter = )] at (B.east)
{
  -1.045&0.265&0.402\\
  -1.045&0.265&0.402\\
  -1.045&0.265&0.402\\
  0.319&-1.260&0.830\\
  0.919&0.670&0.541\\
  0.919&0.670&0.541\\
};

\node(transpose)[right=60pt] at (D.north){{\Large $T$}};

\draw[dashed, thick] (A-1-1.north west) -- (A-1-3.north east);
\draw[dashed, thick] (A-1-1.north west) -- (A-3-1.south west);
\draw[dashed, thick] (A-3-1.south west) -- (A-3-3.south east);
\draw[dashed, thick] (A-1-3.north east) -- (A-3-3.south east);

\draw[dashed, thick] (A-4-4.north west) -- (A-4-4.north east);
\draw[dashed, thick] (A-4-4.north west) -- (A-4-4.south west);
\draw[dashed, thick] (A-4-4.south west) -- (A-4-4.south east);
\draw[dashed, thick] (A-4-4.north east) -- (A-4-4.south east);

\draw[dashed, thick] (A-5-5.north west) -- (A-5-6.north east);
\draw[dashed, thick] (A-5-5.north west) -- (A-6-5.south west);
\draw[dashed, thick] (A-5-6.north east) -- (A-6-6.south east);
\draw[dashed, thick] (A-6-5.south west) -- (A-6-6.south east);
}
\end{tikzpicture}
  \caption{ An illustrative example of low-rank structure of a 3-weakly balanced network.
    The network can be represented as a product of two rank-3 matrices, and so
      the adjacency matrix has rank no more than 3.
  }
  \label{fig:low_rank_example}
\end{figure}

Figure \ref{fig:low_rank_example} is an example of a complete 3-weakly 
balanced network.  As shown, we see its adjacency matrix can be expressed
as a product of two rank-3 matrices, indicating its rank is no more
than three.  In fact, by Theorem \ref{thm:low_rank}, we can conclude
that its adjacency matrix has rank exactly 3.

The above reasoning shows that (adjacency matrices of) complete weakly balanced networks have low rank. 
However, most real networks are not complete graphs. 
Recall that in order to define balance on incomplete networks,
we try to fill in the unobserved or missing edges (relationships) so that balance is obtained (see Definition \ref{def:weak_balance_network}).
Following this desideratum, we can think of sign prediction in signed networks
as a low-rank matrix completion problem. 
Specifically, suppose we observe entries $(i,j) \in \Omega$ of a complete signed network $\comp$.
We want to find a complete matrix by assigning $\pm 1$ to every unknown entry, 
such that the resulting complete graph is weakly balanced and hence, the completed matrix is low rank.  
Thus, our missing value estimation problem can be formulated as:
\begin{align}
  \text{minimize} \  & \ \text{rank}(X) \nonumber\\
  \text{ s.t. } \ & \  X_{ij}=\comp_{ij}, \ \forall \ (i,j)\in 
  \Omega \label{eq:prob1},\\
  & \ X_{ij}\in \{\pm 1\}, \ \forall \ (i,j)\notin \Omega. \nonumber
\end{align}
Once we obtain the minimizer of 
\eqref{eq:prob1}, which we will denote by $\opt$, we can infer
the missing relationship between $i$ and $j$ by simply looking up the sign of the
entry $\opt_{ij}$.  So the question is whether we can solve \eqref{eq:prob1}
efficiently.  In general, \eqref{eq:prob1} is known to be NP-hard; however,
recent research has shown the surprising result that under certain conditions,
the low-rank matrix completion problem~\eqref{eq:prob1} can be solved by convex optimization
to yield a {\em global} optimum in polynomial time~\citep{Candes08a}. In the following subsections, we identify such conditions as well as approaches to approximately solve \eqref{eq:prob1} for real-world signed networks.

\subsection{Sign Prediction via Convex Relaxation}
\label{subsec:MC}
One possible approximate solution for \eqref{eq:prob1} can be obtained by
dropping the discrete constraints and replacing $\text{rank}(X)$ by 
$\|X\|_*$, where $\|X\|_*$ denotes the trace norm of $X$, which is the tightest convex relaxation of rank \citep{Fazel01a}.  Thus, a convex relaxation of~\eqref{eq:prob1} is:
\begin{align}
  \text{minimize} \ &  \ \|X\|_* \nonumber\\
  \text{ s.t. } \ &  \  X_{ij}=\comp_{ij}, \ \forall \ (i,j)\in 
  \Omega. \label{eq:prob2}
\end{align}

It turns out that, under certain condition, by solving~\eqref{eq:prob2} we can recover the {\em exact\/} missing relationships from the underlying complete signed network. This surprising result is the consequence of recent research~\citep{Candes08a, Candes09a} which has shown that perfect recovery  from the observations is possible if the observed entries are uniformly sampled and $\comp$ has high incoherence, which may be defined as follows:
\begin{definition}[Incoherence]
  An $n \times n$ matrix $X$ with singular value decomposition $X=U\Sigma V^T$ is 
  {\it $\mu$-incoherent} if 
  \begin{equation}
    \label{eq:incoherence}
    \max_{i,j} |U_{ij}|\leq \frac{\sqrt{\mu}}{\sqrt{n}} \ \text{ and } \ \max_{i,j} 
    |V_{ij}|\leq \frac{\sqrt{\mu}}{\sqrt{n}}. 
  \end{equation}
  \label{def:incoherence}
\end{definition}

Intuitively, higher incoherence (smaller $\mu$) means that large entries of the matrix  
are not concentrated in a small part. 
The following theorem shows that under high incoherence and uniform sampling,
solving \eqref{eq:prob2} exactly recovers $\comp$ with high probability.

\begin{theorem}[Recovery Condition \citep{Candes09a}]
  Let $\comp$ be an $n \times n$ matrix with rank $k$, with singular value 
  decomposition $\comp=U \Sigma V^T$.  In addition, assume $\comp$ is 
  $\mu$-incoherent.  Then there exists some constant $C$, such that if
  $C\mu^4 n k^2 \log^2 n$ entries are uniformly sampled, 
  then with probability at least $1-n^{-3}$, $\comp$ is the unique optimizer of 
  \eqref{eq:prob2}. 
  \label{thm:mc}
\end{theorem}

In particular, if the underlying matrix has bounded rank (i.e. $k = O(1)$), the
number of sampled entries required for recovery reduces to $O(\mu^4 n \log^2 n)$.

Based on Theorem \ref{thm:mc}, we now show that the notion of incoherence 
can be connected to the relative sizes of the clusters in signed networks.  
As a result,
by solving \eqref{eq:prob2}, we will show that we can recover the underlying signed network with high
probability
if there are no extremely small groups. 
To start, we define the {\it group imbalance} of a signed network as follows:
\begin{definition}[Group Imbalance]
  Let $\comp$ be the adjacency matrix of a complete $k$-weakly balanced
  network with $n$ nodes, and let $n_1, \dots, n_k$ be the sizes of the groups. 
  Group imbalance $\tau$ of $\comp$ is defined as 
  \begin{equation}
    \tau := \max_{i=1,\dots,k} \frac{n}{n_i}.
  \end{equation}
\end{definition}
By definition, $k \leq\tau\leq n$. 
Larger group imbalance $\tau$ indicates 
the presence of a very small group, which would intuitively make
recovery of the underlying network harder~(under uniform sampling).
For example, consider an extreme scenario that a $k$-weakly balanced 
network contains $n$ nodes, with two groups containing only one node each.
Then the adjacency matrix of this network has group imbalance $\tau = n$
with the following form: 
\begin{equation*}
  \comp=
  \begin{bmatrix}[r]
      1 & \cdots & \cdots & -1 & -1 \\
      \vdots & \ddots &      &  \vdots & \vdots \\
      \vdots &        & \ddots& -1 & -1 \\
      -1     & \cdots & -1    &  1 & -1 \\
      -1     & \cdots & -1    &  -1&  1 
    \end{bmatrix}.
\end{equation*}
However, without observing $\comp_{n,n-1}$ or $\comp_{n-1,n}$, it is impossible to
determine whether the last two nodes are in the same cluster, or whether each
of them belongs to an individual cluster.  When $n$ is very large, 
the probability of observing one 
of these two entries will be extremely small.  Therefore,
under uniform sampling of $O(n \log^2 n)$ entries, it is unlikely that any
matrix completion algorithm will be able to exactly recover this network.

Motivated by this example, we now analytically show that 
group imbalance $\tau$ determines the possibility of recovery.  We first show the connection
between $\tau$ and incoherence $\mu$.

\begin{theorem}[Incoherence of Signed Networks]
Let $\comp$ be the adjacency matrix of a complete $k$-weakly balanced network with group imbalance $\tau$. Then $\comp$ is $\tau$-incoherent.
\label{thm:incoherent}
\end{theorem}
\begin{proof}
  Recall from Definition \ref{def:incoherence} that $\mu$ is defined as the maximum absolute value in the (normalized) singular vectors of $\comp$, which are identical to its eigenvectors (up to signs), since $\comp$ is symmetric.
  
  Let $\bu$ be any unit eigenvector of $\comp$ ($\|\bu\|_2=1$) 
  with eigenvalue $\lambda$.  
  Suppose $i$ and $j$ are in the same group, then the $i^{\text{th}}$ and $j^{\text{th}}$ 
  rows of $\comp$ are identical, i.e., $\comp_{i,:} = \comp_{j,:}$.
  As a result, the $i^{\text{th}}$ and $j^{\text{th}}$ elements of all eigenvectors will
  be identical (since $u_i = \comp_{i,:}\bu/\lambda= \comp_{j,:} \bu/\lambda = u_j$). 
  Thus, $\bu$ has the following form: 
  \begin{equation}
    \label{eq:eigen-same}
    \bu = [\underbrace{\alpha_1, \alpha_1, \dots, 
    \alpha_1}_{n_1}, \underbrace{\alpha_2,\dots,\alpha_2}_{n_2},\dots,\underbrace{\alpha_k,\dots,\alpha_k}_{n_k}]^T.
  \end{equation}
Because $\|\bu\|_2=1$,  $\sum_{i=1}^k n_i \alpha_i^2=1$,
  and so $n_i \alpha_i^2\leq 1$, $\forall i$, which implies
  $|\alpha_i|\leq 1/\sqrt{n_i}$, $\forall i$.  Thus,
 \begin{equation*}
    \max_{i} |u_i| = \max_i |\alpha_i| \leq \max_i\frac{1}{\sqrt{n_i}} = 
    \max_i\frac{\sqrt{n/n_i}}{\sqrt{n}} \leq \frac{\sqrt{\tau}}{\sqrt{n}}. 
  \end{equation*}
  Therefore,  $\comp$ is $\tau$-incoherent. 
\end{proof}

Putting together Theorems~\ref{thm:mc} and~\ref{thm:incoherent}, we now have the main theorem of this subsection:
\begin{theorem}[Recovery Condition for Signed Networks]
  \label{thm:recover}
  Suppose we observe edges $\obs_{ij}$, $(i,j) \in \Omega$, from an underlying
  $k$-weakly balanced 
  signed network $\comp$ with $n$ nodes, and suppose that the following assumptions hold:
  \begin{compactenum}[A.]
    \item $k$ is bounded ($k=O(1)$),
    \item the set of observed entries  $\Omega$ is uniformly sampled, and 
\item number of samples is sufficiently large, i.e.
  $|\Omega|\geq C \tau^4 n\log^2 n$, where $\tau$ is the group imbalance of the underlying complete network $\comp$.
\end{compactenum}
Then  $\comp$ can be perfectly recovered by solving \eqref{eq:prob2}, with probability at least $1-n^{-3}$.  
\end{theorem}
In particular, if $n_i/n$ is lower bounded so that $\tau$ is a constant, then 
we only need $O(n\log^2 n)$ observed entries to {\em exactly} recover the complete 
$k$-weakly balanced network. 

\subsection{Sign Prediction via Singular Value Projection}
\label{subsec:SVP}
Though the convex optimization problem \eqref{eq:prob2} mentioned in Section \ref{subsec:MC} 
can be solved to yield the global optimum,
the computational cost of solving it might be too prohibitive in practice.  Therefore, recent research
provides more efficient algorithms to approximately solve \eqref{eq:prob1} \citep{Cai10a, Jain10a}.
In particular, we consider the Singular Value Projection (SVP) algorithm proposed by \cite{Jain10a} which
attempts to solve the low-rank matrix completion problem in an efficient manner.  The SVP algorithm considers
a robust formulation of \eqref{eq:prob1} as follows:
\begin{align}
  \text{minimize} \ & \ \|\proj(\var) - \obs\|^2_F \nonumber\\
  \text{ s.t. } \ &  \  \text{rank}(X) \leq k ,
  \label{eq:SVP_obj}
\end{align}
where the projection operator $\proj$ is defined as:
\begin{equation*}
  (\proj(\var))_{ij} = \begin{cases}
    \var_{ij}, &\text{if } (i,j) \in \Omega \\
    0, &\text{otherwise}.
  \end{cases}
\end{equation*}
Note that the objective \eqref{eq:SVP_obj} recognizes that there might be some
violations of weak balance in the observations $\obs$, and minimizes the 
squared-error instead of trying to enforce exact equality as in \eqref{eq:prob2}.
In an attempt to optimize \eqref{eq:SVP_obj}, the SVP algorithm
iteratively calculates the gradient descent update $\hat{\var}^{(t)}$ of the current solution $\var^{(t)}$, 
and projects $\hat{\var}^{(t)}$ onto the non-convex set of matrices whose rank $\leq k$ using SVD.  After the optimal $\opt$ 
of \eqref{eq:SVP_obj} is derived, one can take the sign of each entry of $\opt$ to obtain an
approximate solution of \eqref{eq:prob1}.
The SVP procedure for sign prediction is summarized in Algorithm \ref{alg:SVP}.

\begin{algorithm}
  \caption{Sign Prediction via Singular Value Projection (SVP)}
  \label{alg:SVP}
  \KwIn{Adjacency matrix $\obs$, rank $k$, tolerance $\epsilon$, max iteration $t_{\max}$, step size $\eta$}
  \KwOut{$\opt$, the completed low-rank matrix that approximately solves \eqref{eq:prob1}}
  \begin{compactenum}
    \item Initialize $\var^{(0)} \leftarrow 0$ and $t \leftarrow 0$.
    \item Do
    \begin{compactitem}
	\item $\hat{\var}^{(t)} \leftarrow \var^{(t)} - \eta(\proj(\var^{(t)}) - \obs)$
	\item $[U_k, \Sigma_k, V_k] \leftarrow$ Top $k$ singular vectors and singular values of $\hat{\var}^{(t)}$
	\item $\var^{(t+1)} \leftarrow U_k \Sigma_k {V_k}^T$
	\item $t \leftarrow t+1$
    \end{compactitem}
    while $\|\proj(\var^{(t)}) - A\|_F^2 > \epsilon$ and $t < t_{\max}$
    \item $\opt \leftarrow \sign(\var^{(t)})$
\end{compactenum}
\end{algorithm}

In addition to its efficiency, experimental evidence provided by
\cite{Jain10a} suggests that if observations are uniformly distributed,
then all iterates of the SVP algorithm are $\mu$-incoherent, 
and if this occurs, then it can be shown that 
the matrix completion problem \eqref{eq:prob1} can be exactly solved by SVP. 
In Section \ref{sec:experiment}, we will see 
that SVP performs well in recovering weakly balanced networks.

\subsection{Sign Prediction via Matrix Factorization}
\label{subsec:MF}
A classical limitation of both convex relaxation and SVP is that 
they require uniform sampling
to ensure good performance.
However, this assumption is violated in most real-life
applications, and so these approaches do not work very well in practice.
In addition, both methods cannot scale to very large datasets.
Thus, we use a gradient based matrix factorization approach as an approximation to the signed network completion problem. In Section \ref{sec:experiment}, we will see that this matrix
factorization approach can boost the accuracy of estimation as well as scale to large real networks.

In the matrix factorization approach, we consider the following problem: 
 \begin{equation}
  \min_{W,H\in \real^{n \times k}} \sum_{(i,j)\in \Omega} 
  (\obs_{ij}-(WH^T)_{ij})^2 + \lambda \|W\|_F^2 + \lambda\|H\|_F^2. 
  \label{eq:mf_prob1}
\end{equation}

Although problem~\eqref{eq:mf_prob1} is non-convex, 
it is widely used in practical collaborative filtering applications as
the performance is competitive or better as compared to trace-norm minimization, while 
scalability is much better. For example, to solve the Netflix problem, 
\eqref{eq:mf_prob1} has been applied with a fair amount of success to factorize datasets with  
100 million ratings \citep{Koren09a}.  

Nevertheless, there is an issue when modeling signed networks using \eqref{eq:mf_prob1}:
the squared loss in the first term of \eqref{eq:mf_prob1} tends to force 
entries of $WH^T$ to be either $+1$ or $-1$.  However,
what we care about in this completion task is the consistency between $\sign((W H^T)_{ij})$
and $\sign(\obs_{ij})$ rather than their difference.
For example, $(WH^T)_{ij}=10$ should have zero loss when $\obs_{ij}=+1$ 
if only the signs are important. 

To resolve this issue, instead of using the squared loss, we use a loss function that only penalizes
the inconsistency in {\it sign}.  More precisely, objective \eqref{eq:mf_prob1} 
can be generalized as:
\begin{equation}
  \min_{W,H\in \real^{n \times k}} 
  \sum_{(i,j)\in \Omega} \loss (\obs_{ij}, (WH^T)_{ij}) + \lambda \|W\|_F^2 + 
  \lambda \|H\|_F^2. 
  \label{eq:mf_prob2}
\end{equation}
In order to penalize inconsistency of sign, we can change the loss function to be the sigmoid or squared-hinge loss:
\begin{align}
  \loss_{\text{sigmoid}}(x, y) &= 1/(1+\exp(xy)), \nonumber \\
  \loss_{\text{square-hinge}}(x, y) &= (\max(0, 1-xy))^2.
  \label{eq:other_loss}
\end{align}
In Section \ref{sec:experiment}, we will see that applying 
sigmoid or squared-hinge loss functions slightly improves prediction accuracy.

\subsection{Time Complexity of Sign Prediction Methods}
There are two main optimization techniques for 
solving \eqref{eq:mf_prob2} for large-scale data: Alternating Least Squares (ALS) and 
Stochastic Gradient Descent (SGD) \citep{Koren09a}. ALS solves the squared loss problem 
\eqref{eq:mf_prob1} by alternately minimizing $W$ and $H$. 
When one of $W$ or $H$ is fixed, the optimization problem becomes a 
least squares problem with 
respect to the other variable, so that we can use well developed least squares 
solvers to solve each subproblem. 
Given an $n \times n$ observed matrix with $m$ observations, it requires 
$O(mk^2)$ operations to form the Hessian matrices, 
and $O(nk^3)$ operations to solve each least squares subproblem. Therefore, the time complexity of ALS is 
$O(t_1(mk^2+nk^3))$ where $t_1$ is the number of iterations.

However, ALS can only be used when the loss function is the squared loss. To solve 
the general form \eqref{eq:mf_prob2} with various loss functions, we use
stochastic gradient descent~(SGD).
In SGD, for each iteration, we pick an observed entry $(i,j)$ at random, 
and only update the $i^{\text{th}}$ row $\bw^T_i$ of $W$ and the $j^{\text{th}}$ 
row $\bh^T_j$ of $H$.
The update rule for $\bw^T_i$ and $\bh^T_j$ is given by:
\begin{align}
  \label{eq:sgd_update}
  \bw^T_i \leftarrow \bw^T_i - \eta \left(\frac{\partial \loss(\obs_{ij},(W 
  H^T)_{ij})}{\partial \bw^T_i} + \lambda \bw^T_i\right), \nonumber \\
  \bh^T_j \leftarrow \bh^T_j - \eta \left(\frac{\partial \loss(\obs_{ij},(W 
  H^T)_{ij})}{\partial \bh^T_j} + \lambda \bh^T_j\right),
\end{align}
where $\eta$ is a small step size.
Each SGD update costs $O(k)$ time, and the total cost of sweeping through all
the entries is $O(mk)$.
Therefore, the time complexity for SGD is $O(t_2mk)$, where $t_2$ is 
the number of iterations taken by SGD to converge. Notice that although the 
complexity of SGD is linear in $k$, it usually takes many more iterations to 
converge compared with ALS, i.e., $t_2 > t_1$.  

On the other hand, all cycle-based algorithms introduced in Section \ref{sec:longer_cycles}
require time at least $O(nm)$,
because they involve
$n \times n$ sparse matrix multiplication steps in model construction. 
In particular, in case of the most effective cycle-based method HOC,
for features with length $\ell$, the number of features is exponential in $\ell$
even if we reduce number of features by ignoring the directions 
(see Section \ref{sec:reduce_features} for details).
Therefore, 
  the time complexity for HOC methods will be $O(2^{\ell}nm)$, which
  is much more expensive than both ALS and SGD as shown in  
Table \ref{tab:complexity} (note that in real large-scale
social networks, $m > n \gg t_1, t_2, k$).

\begin{table}
\centering
\begin{tabular}{c|c|c}
  HOC & LR-ALS & LR-SGD  \\
  \hline
  $O(2^{\ell}nm)$ & $O(t_1(nk^3+mk^2))$ & $O(t_2 km)$\\ 
\end{tabular}
\caption{Time complexity of cycle-based method (HOC) and low rank modeling methods (LR-ALS, LR-SGD).  The HOC time only considers feature computation time.  The time for low rank modeling consists of total model construction time.}
\label{tab:complexity}
\end{table}

\subsection{Clustering Signed Networks}
\label{sec:clustering}
In this section, we see how to take advantage of the low-rank structure of
signed networks to find clusters.  
Based on weak balance theory, the general goal of clustering for signed graphs
is to find a $k$-way partition such that most within-group edges are positive and
most between-group edges are negative. 
One of the state-of-the-art algorithms for clustering signed networks, proposed by \cite{Kunegis10a},
extends spectral clustering by using the signed Laplacian matrix.
Given a
partially observed signed network $\obs$, the signed Laplacian $\bar{L}$ is defined
as $\bar{D}-\obs$, where $\bar{D}$ is a diagonal matrix such that $\bar{D}_{ii} = \sum_{j\neq i}|\obs_{ij}|$.
By this definition, the clustering of signed networks can be
derived by computing the top $k$ eigenvectors
of $\bar{L}$, say $U \in \real^{n \times k}$, 
and subsequently running the $k$-means algorithm on $U$ 
to get the clusters.  This procedure is analogous to the standard spectral
clustering algorithm on unsigned graphs; the only difference being that the usual graph
Laplacian is replaced by the signed Laplacian.



\begin{algorithm}[t]
  \caption{Clustering with Matrix Completion}
  \label{alg:clustering}
  \KwIn{Adjacency matrix $\obs$, number of clusters $k$}
  \KwOut{Cluster indicators}
  \begin{compactenum}
    \item
  $\opt \leftarrow \text{Completion}(\obs)$ with any matrix completion algorithm.  
\item
  $U \leftarrow \text{Top $k$ eigenvectors of $\opt$}$.
\item
  Run any feature-based clustering algorithm on $U$.
\end{compactenum}
\end{algorithm}

However, there is no theoretical guarantee that the use of the signed Laplacian can recover the true groups in a weakly-balanced signed network.
To overcome this theoretical defect,  we now give an algorithm which, under certain conditions, is able to recover the real structure 
even with partial observations.  The key idea is as follows. Since, in Theorem \ref{thm:low_rank}, we proved 
that $k$-weakly balanced graphs have rank up to $k$, we can obtain good clustering by first running a matrix completion algorithm, say trace-norm minimization, on 
$\obs$. The following theorem shows that the eigenvectors of the completed matrix 
possess a desirable property.
\begin{theorem}
  \label{thm:clustering}
  Let $\obs_{ij}$, $(i,j) \in \Omega$, be entries observed from a complete
  $k$-weakly balanced network $\comp$ with $n$ nodes,
  and assume that the solution of \eqref{eq:prob2} is $\opt$ with eigenvectors 
  $U=[\bu_1, \bu_2, \cdots, \bu_k]$.  If the assumptions in Theorem \ref{thm:recover} 
  are all satisfied, then $U_{i,:}=U_{j,:}$ iff 
  $i$ and $j$ are in the same cluster in $\comp$
  with probability at least $1-n^{-3}$.
\end{theorem}
\begin{proof}
  From Theorem \ref{thm:recover}, we know the recovered matrix $\opt$ will be $\comp$
  with probability $\geq 1-n^{-3}$ if the assumptions hold.  
  Suppose $\bu_1,\dots,\bu_k$ are the $k$ eigenvectors of $\opt$.
  From the proof of Theorem \ref{thm:incoherent}, the eigenvectors will have 
  the form in~\eqref{eq:eigen-same}, which means $U_{i,:}=U_{j,:}$ if $i$ and 
  $j$ are in the same cluster. 
  Furthermore, when $i$ and $j$ are in different clusters, 
  $\comp_{i,:} \neq \comp_{j,:}$, so $U_{i,:}$ cannot equal $U_{j,:}$. 
  This proves the theorem. 
\end{proof}

Following this theorem, the true clusters can be identified from the eigenvectors of $\opt$ when
the assumptions in Theorem \ref{thm:recover} hold.  Therefore, perfect clustering 
is guaranteed in this scenario. 

More generally, we can use any matrix completion method 
to complete $\obs$.  For example, if we take SVP
as the matrix completion approach, we can obtain a perfect clustering 
result if all iterates of the algorithm are $\mu$-incoherent.
Under the latter condition, SVP can recover $\comp$ exactly, so
the property of eigenvectors in Theorem \ref{thm:clustering} can again be used.
Our clustering algorithm that uses matrix completion is summarized in Algorithm~\ref{alg:clustering}. 

It should not be surprising that our clustering algorithm is superior to
(signed) spectral clustering.  In some sense, our approach can be viewed 
as a spectral method, except that it first fills in the missing links from
the training data by doing matrix completion.  This step is simple yet crucial
in signed networks as it overcomes the sparsity of the network.
We will see that our clustering algorithm outperforms the 
(signed) spectral clustering method in Section~\ref{sec:experiment}.



\section{Experimental Results}
\label{sec:experiment}
We now present experimental results for sign prediction and clustering using our proposed methods.
For sign prediction, we show that local methods, such as MOI and HOC, yield better
predictive accuracy if longer cycles are considered.  In addition, 
if we consider the global low-rank structure of the network, 
   prediction via matrix factorization further outperforms local methods
in terms of both accuracy and running time.  
For clustering, we show that clustering via low rank
model gives us better results than clustering via signed Laplacian.  These results suggest
the usefulness of the global perspective of social balance.

\subsection{Description of Data Sets}
In our experiments, we consider both synthetic and real-life datasets. 
To construct synthetic networks, 
we first consider a complete $k$-weakly balanced network $\comp$, and sample
some entries from $\comp$ to form the partially observed network $\obs$,
with three controlling parameters: \textit{sparsity} $s$, \textit{noise level} $\epsilon$
and \textit{sampling process} $\dist$.
The sparsity $s$ controls the percentage of edges we sample from $\comp$.
The noise level $\epsilon$ specifies the probability that the sign of a sampled
edge is flipped. The sampling process $\dist$ specifies how the sampled entries
are distributed.  In particular, we will focus on two sampling distributions: 
uniform and power-law distribution, denoted as $\uni$ and $\power$ respectively.
Thus, a partially observed network $\obs$ can be described as 
$\obs = \comp(s,\epsilon,\dist)$.

We also consider three real-life signed networks: Epinions, Slashdot, Wikipedia.
Epinions is a consumer review network in which users can either trust or distrust other consumer's reviews.
Slashdot is a discussion web site in which users can recognize others as friends or enemies.
Wikipedia is a who-vote-to-whom network in which users can vote for or against others 
to be administrators in Wikipedia.
These three datasets have previously been used as benchmarks for sign prediction \citep{Leskovec10a,Chiang11a}. 
Table \ref{tab:real_data} shows the statistics of these real signed networks.

\begin{table}
\centering
\caption{Network Statistics}
\label{tab:real_data}
\begin{tabular}{c|cccc}
& \# nodes& \# edges & + edges & - edges \\
    \hline
    Wikipedia & 7,065 & 103,561 & 78.8\% & 21.2\%  \\
      Slashdot & 82,144  & 549,202 & 77.4\% & 22.6\% \\
      Epinions & 131,828  & 840,799 & 85.0\% & 15.0\% 
      \end{tabular}
\end{table}

\subsection{Evidence of Local and Global Patterns in Real Signed Networks}

We have seen that cycles in
signed networks exhibit structural balance (see Definition
\ref{def:balanced_cycle}) according to balance theory, and that we can make use of cycles for predictions 
(see Sections \ref{sec:triads} and \ref{sec:longer_cycles}). Indeed, cycles tend to be balanced in real-life networks. In all three real networks we consider, \cite{Leskovec10b}
found that balanced triads (i.e. 3-cycles) are much more likely to be observed than unbalanced triads. Our study supports that the local patterns (i.e. $\ell$-cycles) of the three networks tend to be balanced. For each network $A$, we consider all patterns of 3-cycles and 4-cycles  in the symmetric network $\sign(A+A^T)$. For convenience, we use $C_{\ell i}$ to denote the $i^{\text{th}}$ pattern of an $\ell$-cycle. Thus, we simply consider all $C_{\ell i}$ for $\ell = 3, 4$ for each network. The patterns of these cycles are shown in Table~\ref{tab:local_evidence}. We first calculate the probability that the configuration of a given $\ell$-cycle is $C_{\ell i}$, denoted $P(C_{\ell i})$.  We then randomly shuffle the
sign of edges in the network and calculate the same probability on the shuffled network, which is denoted $P_0(C_{\ell i})$. Thus $P_0(C_{\ell i})$ can be viewed as the (expected) probability that $C_{\ell i}$ is observed if the sign of edges has no particular pattern.  
With the two probabilities, we calculate the ``surprise" of $C_{\ell i}$ that measures how significantly $C_{\ell i}$ appears more or less than expected.  Formally, the surprise of
$C_{\ell i}$, denoted $S(C_{\ell i})$, is defined as:
\begin{equation*}
  S(C_{\ell i}) := \frac{\Delta_\ell P(C_{\ell i}) - \Delta_\ell P_0(C_{\ell i})}{\sqrt{\Delta_\ell P_0(C_{\ell i})(1-P_0(C_{\ell i})})},
\end{equation*}
where $\Delta_\ell$ is number of $\ell$-cycles in the network.  
The above quantity is basically the number of standard deviations that 
the observed value of $C_{\ell i}$ differs from the expected value of $C_{\ell i}$ in the shuffled network.  
See \citep{Leskovec10b} for more discussions.

Table \ref{tab:local_evidence} shows the real probability, the expected probability,
and the surprise value of each $C_{\ell i}$ in three networks.  From the surprise values,
we observe that
cycles with all positive edges (i.e. $C_{31}$, $C_{41}$) are far more than expected, 
and cycles with one negative edge (i.e. $C_{32}$, $C_{42}$) are far less than
expected.  The observations suggest that the number of both balanced/unbalanced patterns are
significantly larger/smaller than expected when the cycle contains many 
positive edges.  Readers might notice that some balanced cycles have
large negative surprise values (for example, $C_{44}$ in Epinions).
However, in both real and shuffled networks, the fraction of such cycles 
are actually quite similar.  The negative value of surprise is 
amplified by large number of observations of $\ell$-cycles
(for example, $\Delta_4 = 6.85 \times 10^8$).  Furthermore, we also
calculate these statistics on all balanced 3 and 4-cycles as shown in the last
two rows in Table \ref{tab:local_evidence}.  Both the
difference between $P(C)$ and $P_0(C)$ and the surprise value of balanced cycles
are quite large.  Overall, we find that local balanced patterns are somewhat significant.

\begin{table}
\centering
\caption{Statistics of balanced and unbalanced $\ell$-cycles, $\ell = 3, 4$
  (notice that $\sum_i P(C_{\ell i}) = \sum_i P_0(C_{\ell i}) = 1$ due to the property of probability).
The first 6 cycles are balanced and the last 4 cycles are unbalanced.  
The last two rows show that balanced 3-cycles and 4-cycles are much more
than expected.}
\label{tab:local_evidence}
{\footnotesize
\begin{tabular}{c|ccr|ccr|ccr}
  
  & \multicolumn{3}{c}{Epinions} & \multicolumn{3}{c}{Slashdot} 
  & \multicolumn{3}{c}{Wikipedia} \\
  Type of cycle 
   & \footnotesize $P(C_{\ell i})$ & \footnotesize$P_0(C_{\ell i})$ & \footnotesize$S(C_{\ell i})$ 
   & \footnotesize $P(C_{\ell i})$ & \footnotesize$P_0(C_{\ell i})$ & \footnotesize$S(C_{\ell i})$ 
   & \footnotesize$P(C_{\ell i})$ & \footnotesize$P_0(C_{\ell i})$ & \footnotesize$S(C_{\ell i})$ 
   \\ \hline

  $C_{31}: +++\quad$ 
   & 0.8259 & 0.5754 & 1107.0
   & 0.7301 & 0.4502 & 425.2
   & 0.6996 & 0.4806 & 335.4 
  \\
  $C_{33}: +--\quad$
   & 0.0791 & 0.0706 & 72.3
   & 0.1364 & 0.1260 & 23.5
   & 0.0840 & 0.1105 & -64.7
  \\
  $C_{41}: ++++$
   & 0.7538 & 0.4777 & 14464.7
   & 0.6723 & 0.3435 & 5120.8
   & 0.6080 & 0.3757 & 3557.6
  \\
  $C_{43}: ++--$
   & 0.0911 & 0.0787 & 1210.6
   & 0.1127 & 0.1286 & -352.1
   & 0.1007 & 0.1155 & -344.1
  \\
  $C_{44}: +-+-$
   & 0.0065 & 0.0393 & -4418.5
   & 0.0138 & 0.0645 & -1528.0
   & 0.0139 & 0.0578 & -1396.4
  \\
  $C_{46}: ----$
   & 0.0103 & 0.0008 & 8722.8
   & 0.0263 & 0.0030 & 3147.7
   & 0.0054 & 0.0022 & 505.4
  \\
  \hline

  $C_{32}: ++-\quad$ 
   & 0.0834 & 0.3493 & -1218.4
   & 0.1125 & 0.4111 & -458.7
   & 0.2052 & 0.3987 & -302.5
  \\
  $C_{34}: --- \quad$
   & 0.0117 & 0.0047 & 220.9
   & 0.0211 & 0.0127 & 56.9
   & 0.0013 & 0.0102 & 8.5
  \\

  $C_{42}: +++-$
   & 0.1174 & 0.3875 & -14508.8
   & 0.1413 & 0.4211 & -4191.5
   & 0.2473 & 0.4167 & -2548.5
  \\
  $C_{45}: +---$
   & 0.0208 & 0.0160 & 1017.7
   & 0.0337 & 0.0392 & -212.0
   & 0.0247 & 0.0320 & -309.3
  \\ \hline \hline
  \footnotesize Balanced 3-cycles
  & 0.9050 & 0.6459 & 1182.9
  & 0.8665 & 0.5763 & 443.9
  & 0.7835 & 0.5911 & 299.6 
  \\
  \footnotesize Balanced 4-cycles 
  & 0.8617 & 0.5965 & 14147.8
  & 0.8250 & 0.5397 & 4234.7
  & 0.7280 & 0.5513 & 2635.6
\end{tabular}
}
\end{table}

On the other hand, in Section \ref{sec:low_rank_model}, we have seen that low rank structure
emerges when we theoretically examine weakly balanced networks.  
We now show that real networks tend to exhibit low-rank structure to 
a much greater extent compared to random networks.
As a baseline, for each real network we create two corresponding random networks for comparison:
the first one is the (symmetric) ER network generated from Erd\"{o}s-R\'{e}nyi 
model \citep{Erdos60a} that preserves the sparsity and the
ratio of positive to negative edges of the compared real network.  The second
one is the shuffled network with the same network structure as the compared real network,
except that we randomly shuffle the sign of edges.

The experiment is conducted as follows.
We first derive the low-rank complete matrix $\opt$
by running matrix completion algorithm on the observed entries $\obs_{ij}$.
Then, we look at the relative error on the observed set $\Omega$:
\begin{equation}
  \text{err}_{\Omega} = \frac{\| W \circ (\opt - \obs) \|_F}{\| \obs \|_F}, 
\end{equation}
where $W_{ij} = 1$ if $(i,j) \in \Omega$ and $W_{ij} = 0$ otherwise, and $\circ$ denotes
element-wise multiplication.  Clearly, smaller $\text{err}_{\Omega}$ indicates better
approximation for the observed entries.

In our experiment, we choose matrix factorization approach for matrix completion, with
ranks $k = 1, 2, 4, 8, 16$ and $32$.  For each network (real networks
and their corresponding random networks), we complete the network with different $k$ values and compute $\text{err}_{\Omega}$.
The result is shown in Figure \ref{fig:real_data}.
Compared to the two random networks, the three real-life
networks achieve much smaller $\text{err}_{\Omega}$ for each small $k$.
This suggests that low-rank matrices provide a better approximation of
the observed entries for real-life networks, as compared to random networks.

\begin{figure*}
\centering
\subfloat[Epinions]{\includegraphics[width=.3\textwidth]{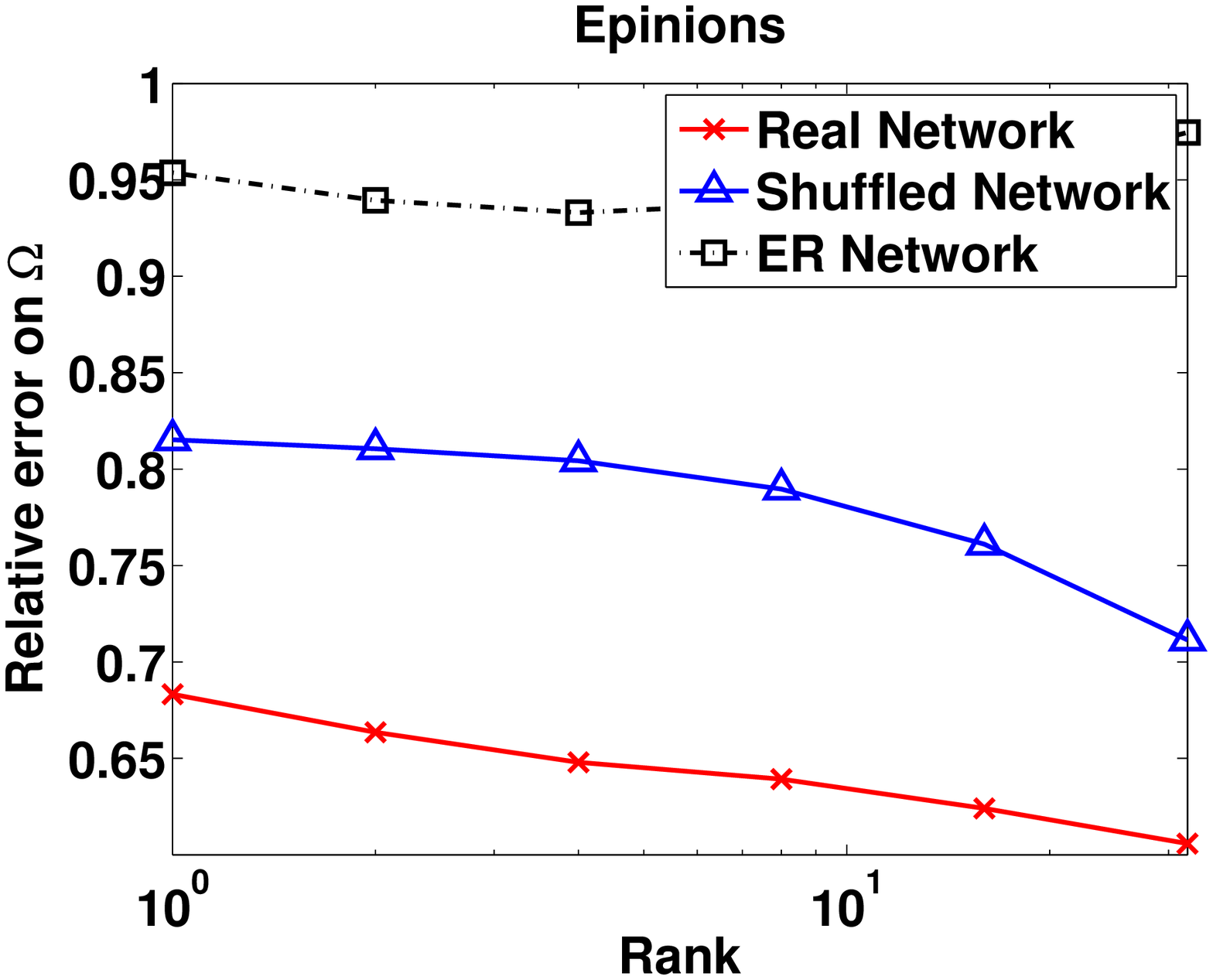}}
\subfloat[Slashdot]{\includegraphics[width=.3\textwidth]{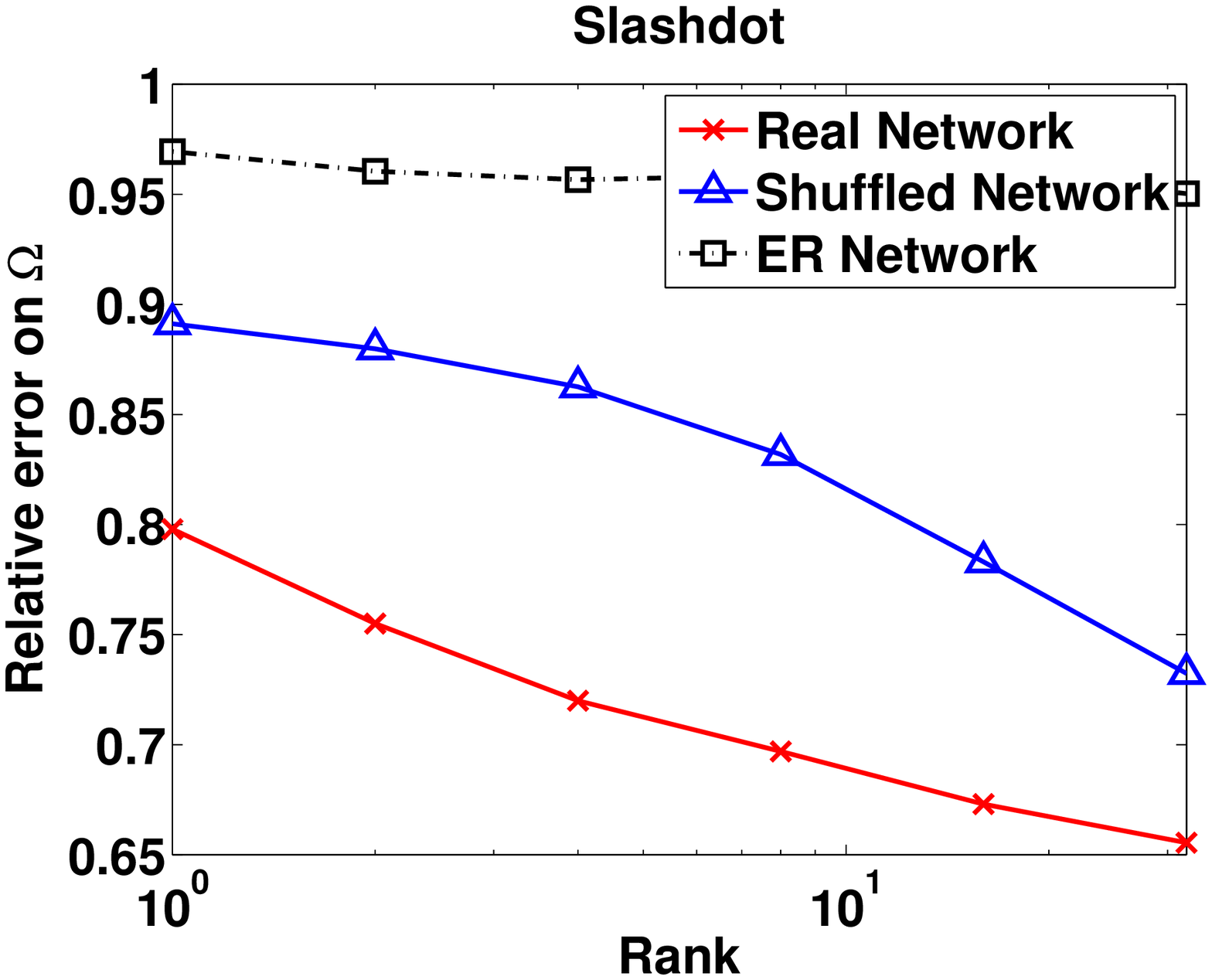}}
\subfloat[Wikipedia]{\includegraphics[width=.3\textwidth]{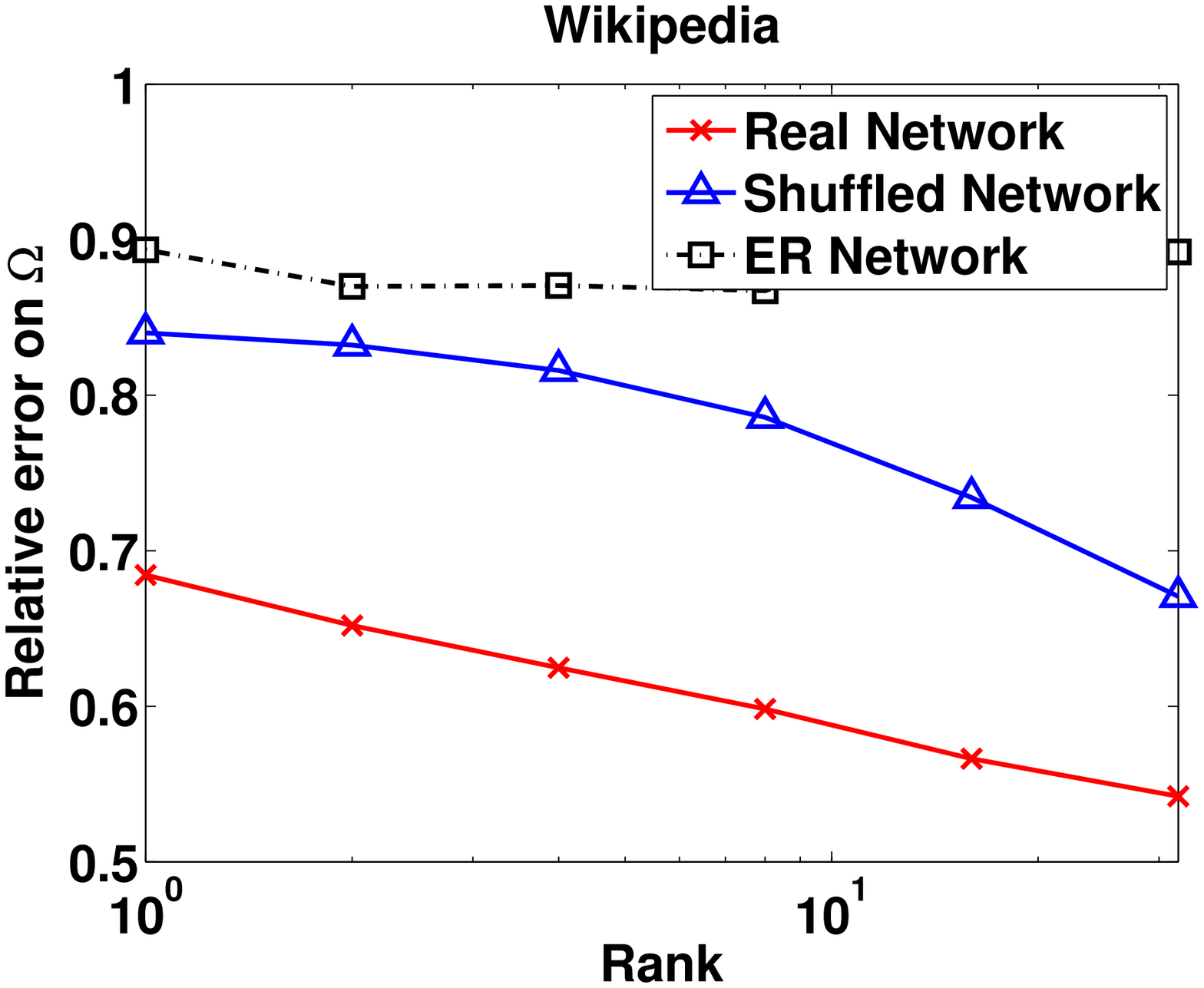}}
  \caption{Relative error on $\Omega$, the observed entries, between adjacency matrix and completed matrix, 
  for real-life networks versus random networks.  Real-life networks achieve much smaller relative error for every
  $k$ as compared with random networks.}
\label{fig:real_data}
\end{figure*}

\subsection{Sign Prediction}
We now compare the performance of our proposed methods for sign prediction.  
As introduced in Sections \ref{sec:triads} and \ref{sec:longer_cycles},
there are two families of cycle-based methods: one based on measures of imbalance (MOI), 
and the other based on the supervised learning using higher order cycles (HOC).
Both families depend on a parameter $\ell \ge 3$
that denotes the order of the cycles that the method is based on.
For MOI, we consider all $\ell$ less than $10$ as well as $\infty$ (recall that in this case
MOI becomes Katz measure), and for HOC we consider $\ell=3,4,5$. Note
that the set of features used by HOC-$(\ell+1)$ is a strict superset of the
features used by HOC-$\ell$. 

We also consider two fully global approaches for low rank matrix completion -- 
Singular Value Projection from Section \ref{subsec:SVP} 
and matrix factorization from Section \ref{subsec:MF}.
The SVP approach (denoted as LR-SVP) is chosen to demostrate that perfect recovery
can be achieved if the observations are uniformly distributed.
For matrix factorization, we consider the ALS method that solves
problem \eqref{eq:mf_prob1}, as well as SGD methods that solve the general problem \eqref{eq:mf_prob2} with sigmoid
loss and square-hinge loss, defined in \eqref{eq:other_loss}.  
We denote these methods as LR-ALS, LR-SIG and LR-SH, respectively.

\subsubsection{Synthetic Datasets}
We first compare all categories of approaches on synthetic datasets.  We choose
LR-SVP, LR-ALS, MOI-$\infty$ and HOC-3 as representatives of the two approaches of low rank matrix completion,  MOI-based,
and HOC-based methods respectively.  
We consider the underlying network $\comp$ to be
a complete $5$-weakly balanced network, where the five clusters have sizes $100$, $200$, $300$, $400$ and $500$.
Instead of observing all of $\comp$, we assume that we only observe a partial network $\obs$ by sampling
some entries from $\comp$ using three sampling procedures:
uniform sampling, uniform sampling with noise, and sampling with power-law distribution.
For each algorithm, we input the observed entries as training data 
and calculate the sign prediction accuracy on the rest of the entries. 

{\bf Uniform sampling:}
In this scenario, we generate several observed networks $\obs = \comp(s, 0, \uni)$.
We vary $s$ from $0.001$ to $0.1$ and plot the prediction accuracy in Figure \ref{fig:nonoise}.
Under this setting, LR-SVP and LR-ALS outperform the cycle-based methods.  
We observe that MOI-$\infty$ performs the worst with accuracy only 50\%-70\%.  
However, if we repeat the same experiment substituting $\comp$ with $\comp_2$, where
$\comp_2$ is a complete strongly balanced network whose two groups have size $1000$,
we observe that MOI and global methods perform alike as shown in Figure \ref{fig:nonoise_2}.
This is because MOI uses cycle-based
measurements to make more cycles become {\it balanced}.  This prediction policy
is most appropriate when $k = 2$ (that is, the underlying network $\comp$ has strong balance), 
but performs poorly when the underlying network is weakly balanced (i.e. more than two groups). 
HOC-3 works much better than 
MOI-$\infty$ since it learns a classifier from cycle-based features
rather than simply making cycles balanced, but its accuracy drops dramatically when
$s$ is less than $0.05$.
On the other hand, both LR-SVP and LR-ALS show high accuracy for all $s \geq 0.01$.
In particular, LR-SVP can achieve 100\% accuracy
when $s>0.07$, which reconfirms the theoretical 
recovery guarantee stated in Theorem \ref{thm:recover}.
Moreover, although LR-ALS has no theoretical guarantee,
it can still recover the ground truth, an observation that is consistent with previous results. 

{\bf Uniform sampling with noise:}
To make the synthetic data more similar to real data, we further add noise into observations. 
We generate observed networks $\obs = \comp(0.1, \epsilon, \uni)$, where $\epsilon$ varies from $0.01$ to $0.25$.  
The result is shown in Figure \ref{fig:noise}.
We can see that global methods are still clearly better than cycle-based methods
when noise level becomes higher.
Moreover,
LR-SVP perfectly recovers $\comp$ when the noise level $\epsilon<0.05$, and
LR-ALS also achieves perfect recovery with a smaller $\epsilon$.

{\bf Sampling with power-law distribution:}
As Sections \ref{subsec:MC} and \ref{subsec:SVP} stated, the exact recovery guarantees of convex relaxation 
and SVP for matrix completion crucially rely on the assumption that 
observed entries are uniformly sampled.  
However, in most real networks (for example, Slashdot in \cite{Kunegis09a}), 
the degree distribution of observed entries follows a power law.
Therefore, we examine how the approaches perform on power-law distributed networks.
The power-law distributed networks are generated
using the Chung-Lu-Vu~(CLV) model proposed by~\cite{Chung03a}, 
which allows one to generate random graphs with arbitrary expected degree 
sequence. 
Similar to the uniform sampling case, we perform the sign prediction task
on $\obs = \comp(s, 0, \power)$ varying $s$ from $0.001$ to $0.1$,
and plot the prediction accuracy in Figure \ref{fig:power}.
We can see that MOI-$\infty$ still has poor performance for weakly balanced graphs.
However, unlike the uniform sampling case, LR-SVP has lower accuracy rate 
compared to HOC-3 when $s < 0.1$.
On the other hand, LR-ALS still performs better than all other methods in power-law distributed graphs. 

\begin{figure*}[th]
  \centering
  \begin{tabular}{cc}
    \subfloat[Uniformly sampled without noise ($k = 5$)\label{fig:nonoise}]{\includegraphics[width=0.4\textwidth,height=5cm]{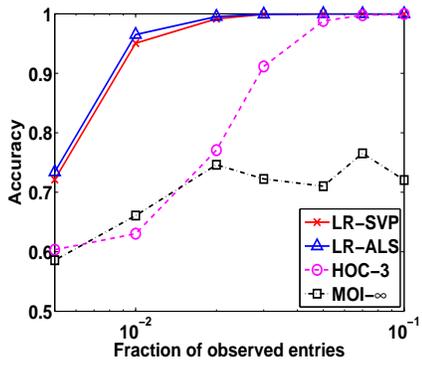}}
 &  
    \subfloat[Uniformly sampled without noise on balanced networks ($k = 2$)\label{fig:nonoise_2}]{\includegraphics[width=0.4\textwidth,height=5cm]{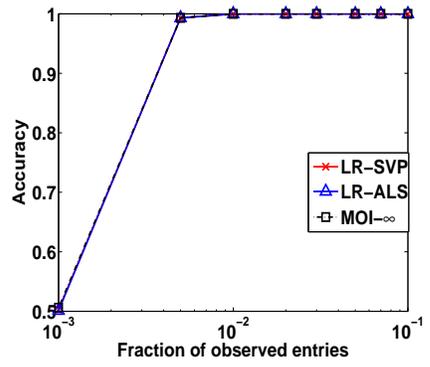}}
    \\
    \subfloat[Uniformly sampled with noise ($k = 5$)\label{fig:noise}]{\includegraphics[width=0.4\textwidth,height=5cm]{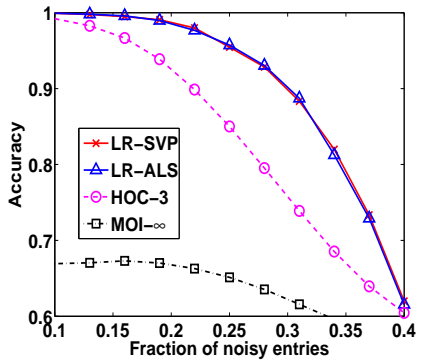}}
 &
 \subfloat[Power-law distributed networks ($k = 5$)\label{fig:power}]{\includegraphics[width=0.4\textwidth,height=5cm]{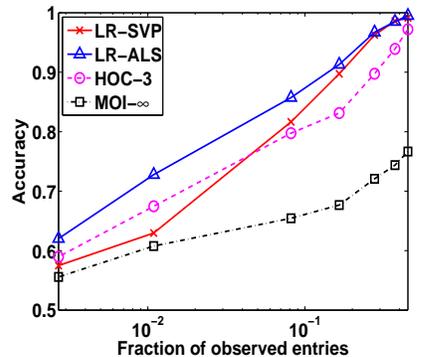}}
  \end{tabular}
  \caption{
  Sign prediction accuracy of local and global methods on synthetic datasets.
    On (strongly) balanced networks (\ref{fig:nonoise_2}), MOI-$\infty$ is seen to perform
    as well as LR-SVP and LR-ALS.  However, in general weakly balanced networks, 
    global methods LR-SVP and LR-ALS outperform cycle-based methods such as MOI-$\infty$ and HOC-3.
  In addition, LR-ALS is more robust than LR-SVP when the observations 
  are sampled from a power-law distribution.}
  \label{fig:synthetic}
\end{figure*}


From results on synthetic data shown in Figure \ref{fig:synthetic}, 
we can conclude that global methods generally do better than local methods
on synthetic setting, and the low rank model with matrix factorization
(LR-ALS) performs the best in most cases, even when observed entries are not uniformly distributed.

\subsubsection{Real-life Datasets}
Now we further evaluate our sign prediction methods on three real-life networks.
To begin with, we evaluate and compare MOI methods using a {\em leave-one-out} type methodology: 
each edge in the network is successively removed and the method tries to predict the sign of that edge using the rest of the network. 
Figure \ref{fig:MOI} shows the accuracy of MOI based methods. Note that the accuracy is shown for edges with embeddedness under a certain threshold. First, we see that the accuracy is a non-decreasing function of the embeddedness threshold. 
Next, it is clear that higher-order methods perform significantly better than MOI-$3$ (triangles) method. 
Finally, the performance boost is larger for edges with low embeddedness. This is expected as edges of low embeddedness by definition do not have many common neighbors for their end-points, and higher-order cycles have relatively better information for such edges than others. 
We also observe from our experiments that beyond $\ell = 5$, the performance gain is not very significant. 

\begin{figure*}
\centering
\subfloat[Epinions]{\includegraphics[width=.3\textwidth]{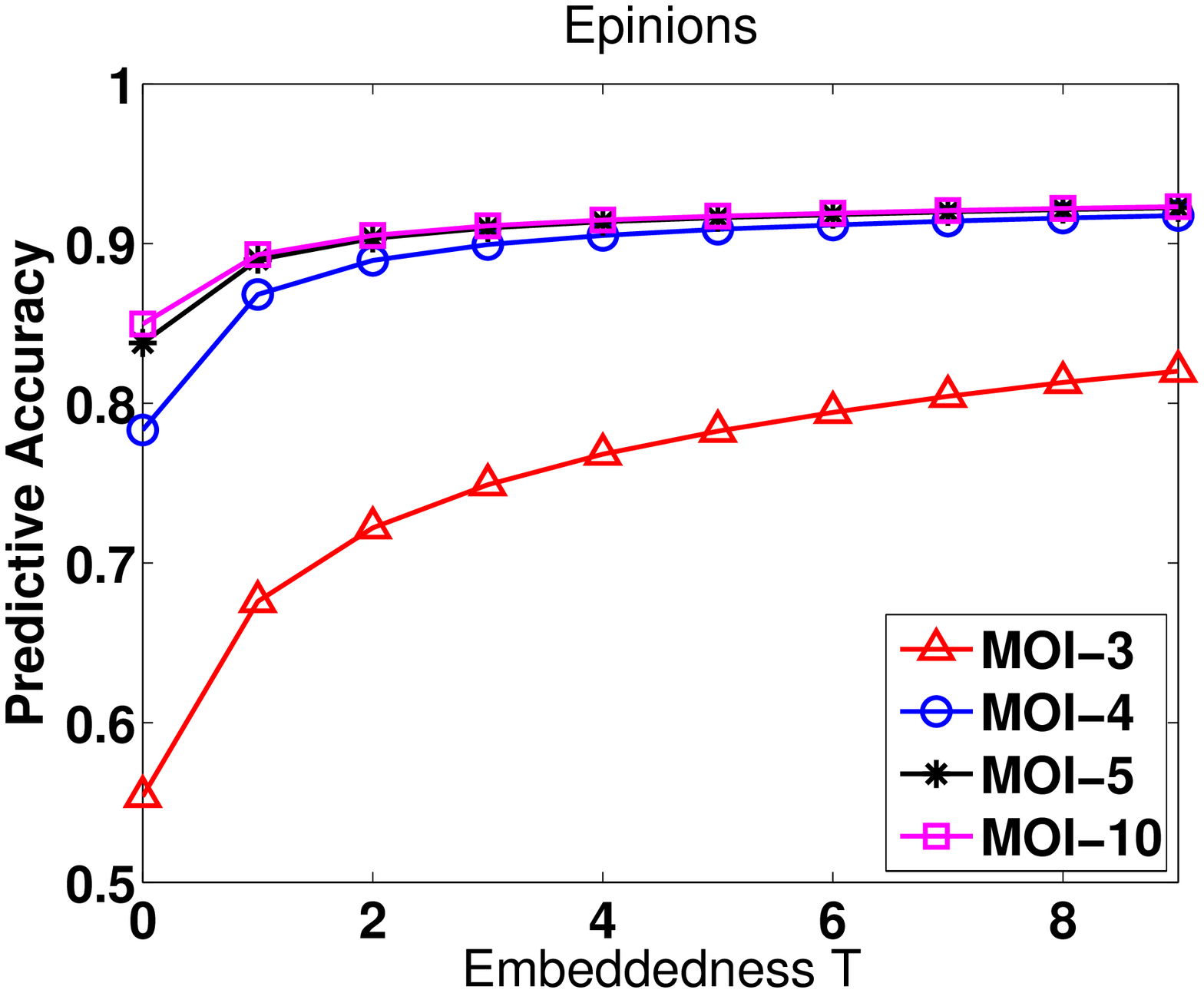}}
\subfloat[Slashdot]{\includegraphics[width=.3\textwidth]{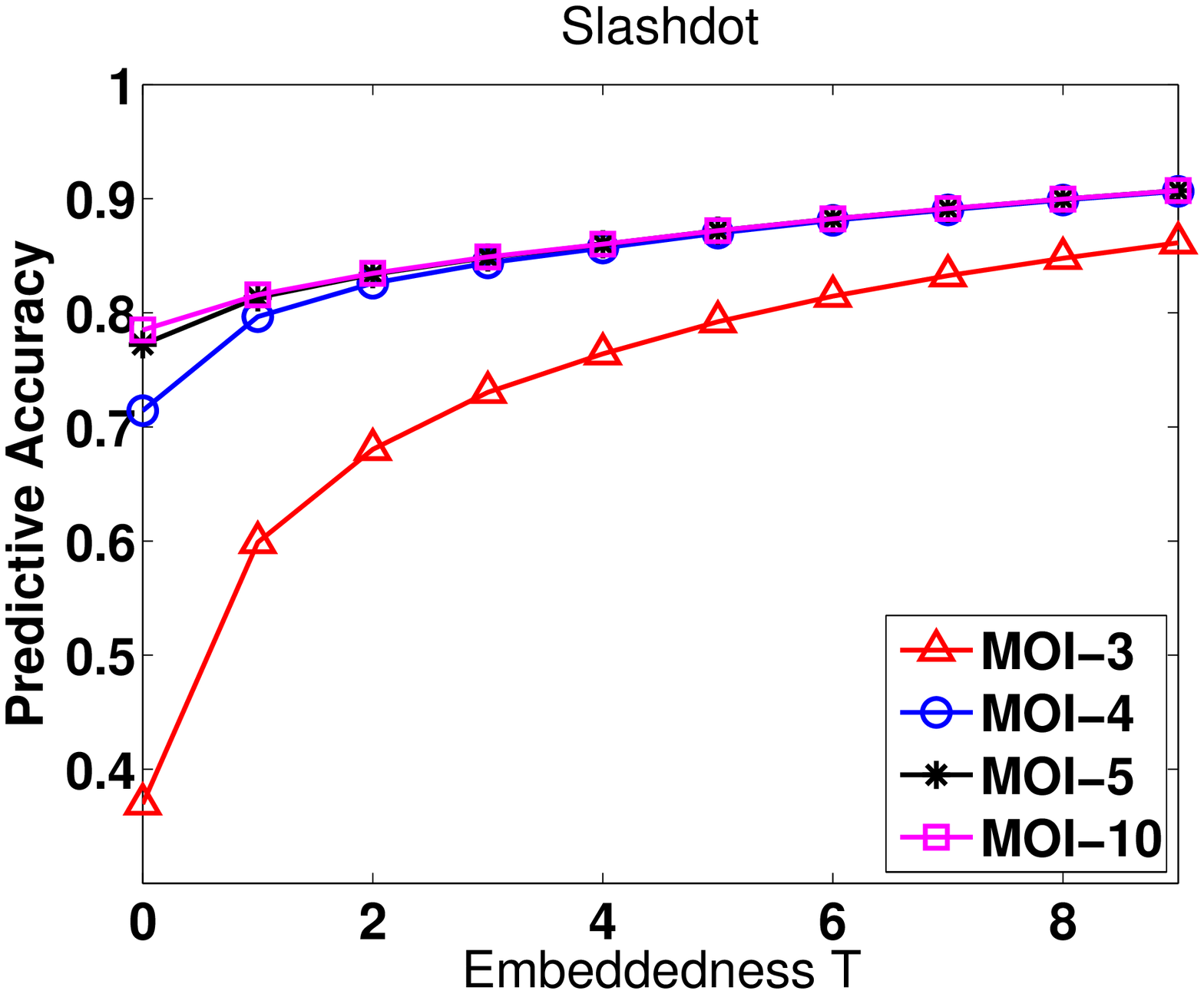}}
\subfloat[Wikipedia]{\includegraphics[width=.3\textwidth]{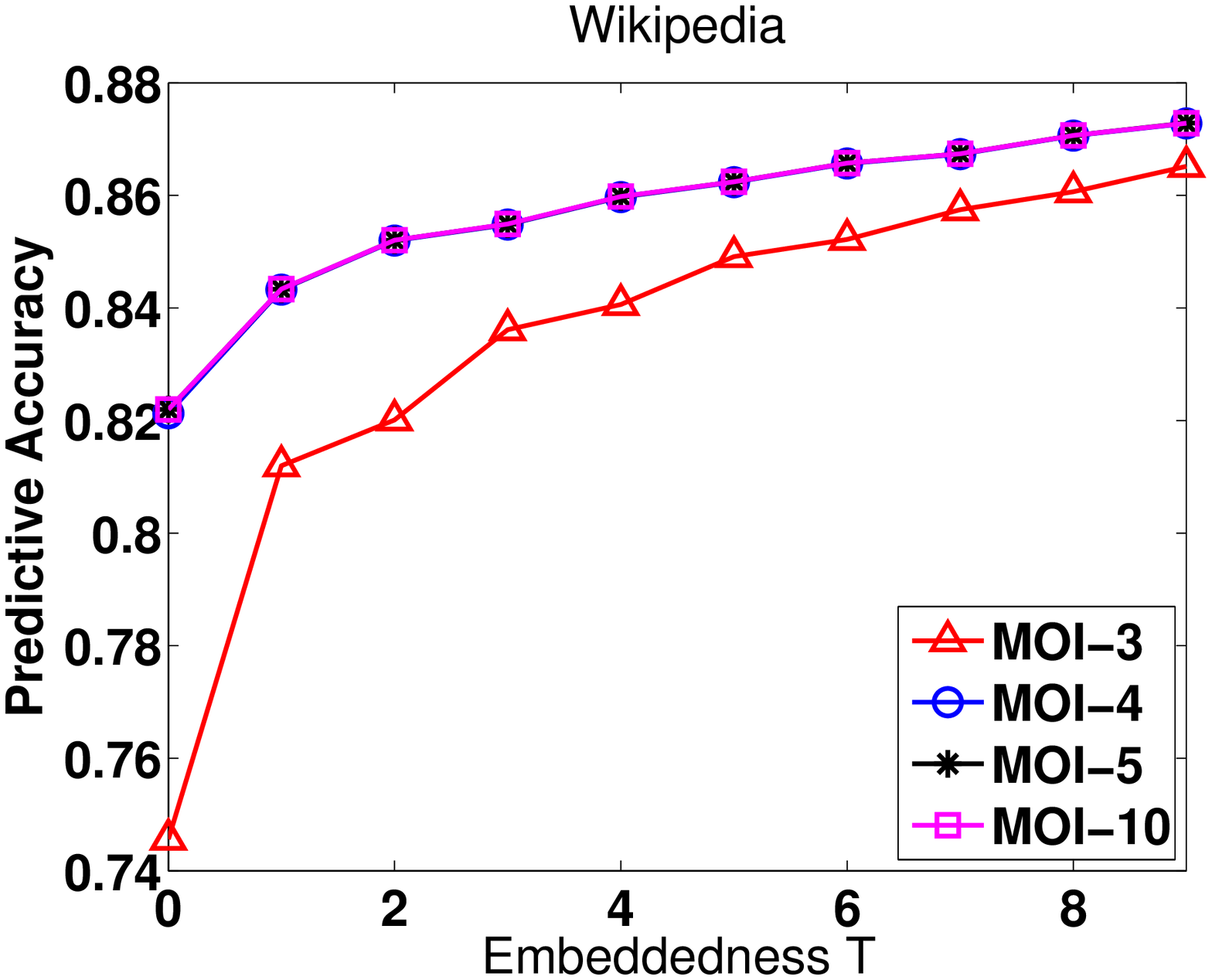}}
\caption{Accuracy of Measures of Imbalance (MOI) Based Methods for $\ell=3, 4, 5, 10.$ \textnormal{These plots show the accuracy
of MOI-$\ell$ methods for edges with embeddedness {\em at least} $T$ for various thresholds $T$. We see that the difference in the performance
of MOI-$3$ and higher order methods is larger when edges with lower embeddedness are considered. We also see that the improvement
obtained by going beyond order $5$ is not very significant.} }
\label{fig:MOI}
\end{figure*}

Next, we compare HOC methods.  We resort to {\em $10$-fold cross-validation}.  To be more concrete, we (randomly) created $10$ disjoint test folds each consisting of $10\%$ of the total number of edges in the network. 
For each test fold, the remaining $90\%$ of the edges serve as the training set. 
For a given test fold, the feature extraction and logistic model training happen on a graph with the test edges removed. 
To evaluate HOC methods, we consider not only prediction accuracies but also false-positive rates.
We report both accuracies and false-positive rates by averaging them over the $10$ folds. 
As shown in Table \ref{tab:prediction_acc},
in all the datasets, there is a small improvement in accuracy by using higher order cycles (HOC-$5$).
The false positive rate, however, reveals a more interesting phenomenon in Figure~\ref{fig:HOC}. 
Indeed, higher order methods (such as HOC-$5$) significantly reduce the false positive rate as compared to that of HOC-$3$. 
However Figure~\ref{fig:HOC} shows that, unlike MOI based methods, edge embeddedness does not seem to affect the decrease in false positive rate for HOC methods. We see this trend across all the datasets.

\begin{figure}
\centering
\subfloat[Epinions]{\includegraphics[width=.3\textwidth]{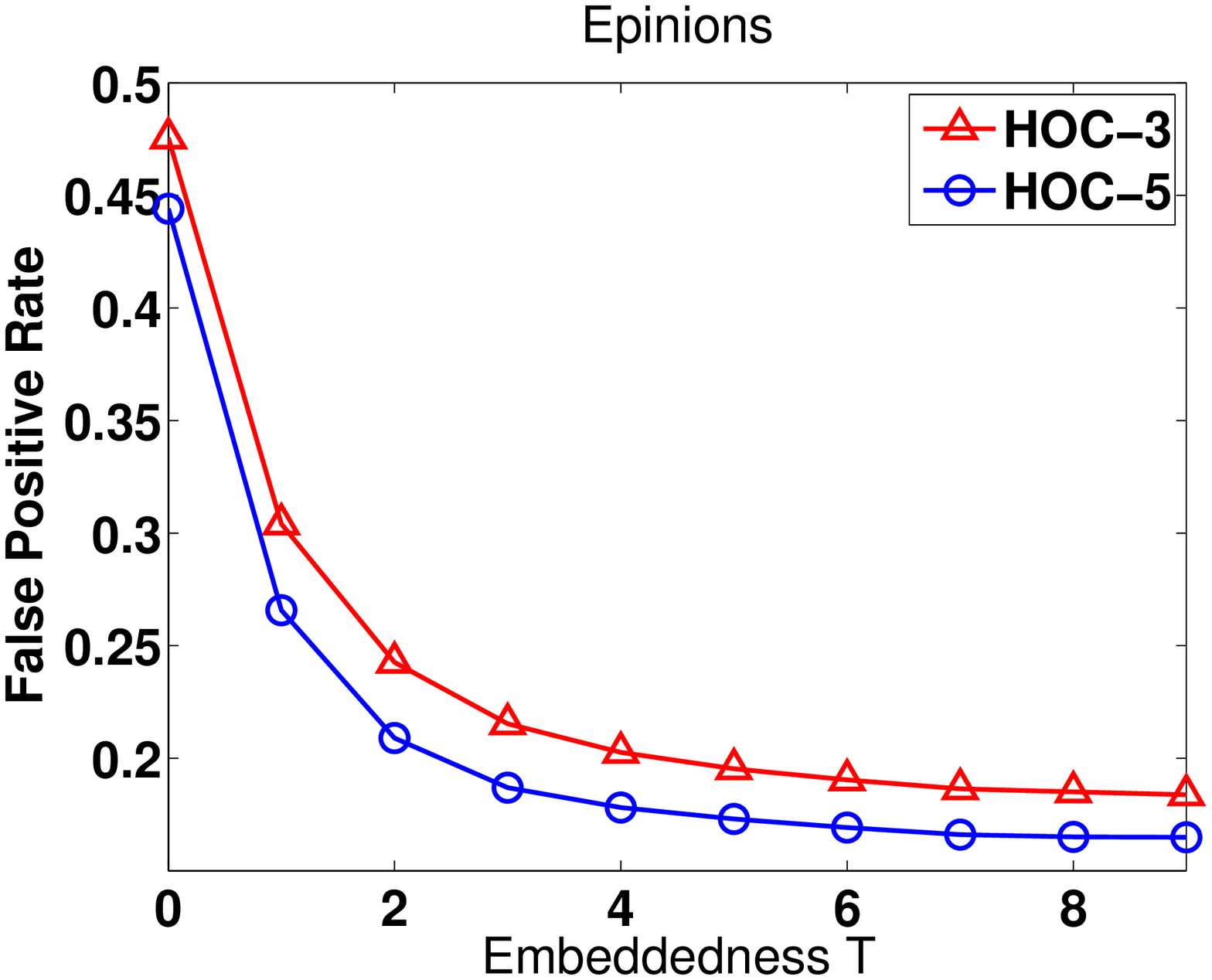}}
\subfloat[Slashdot]{\includegraphics[width=.3\textwidth]{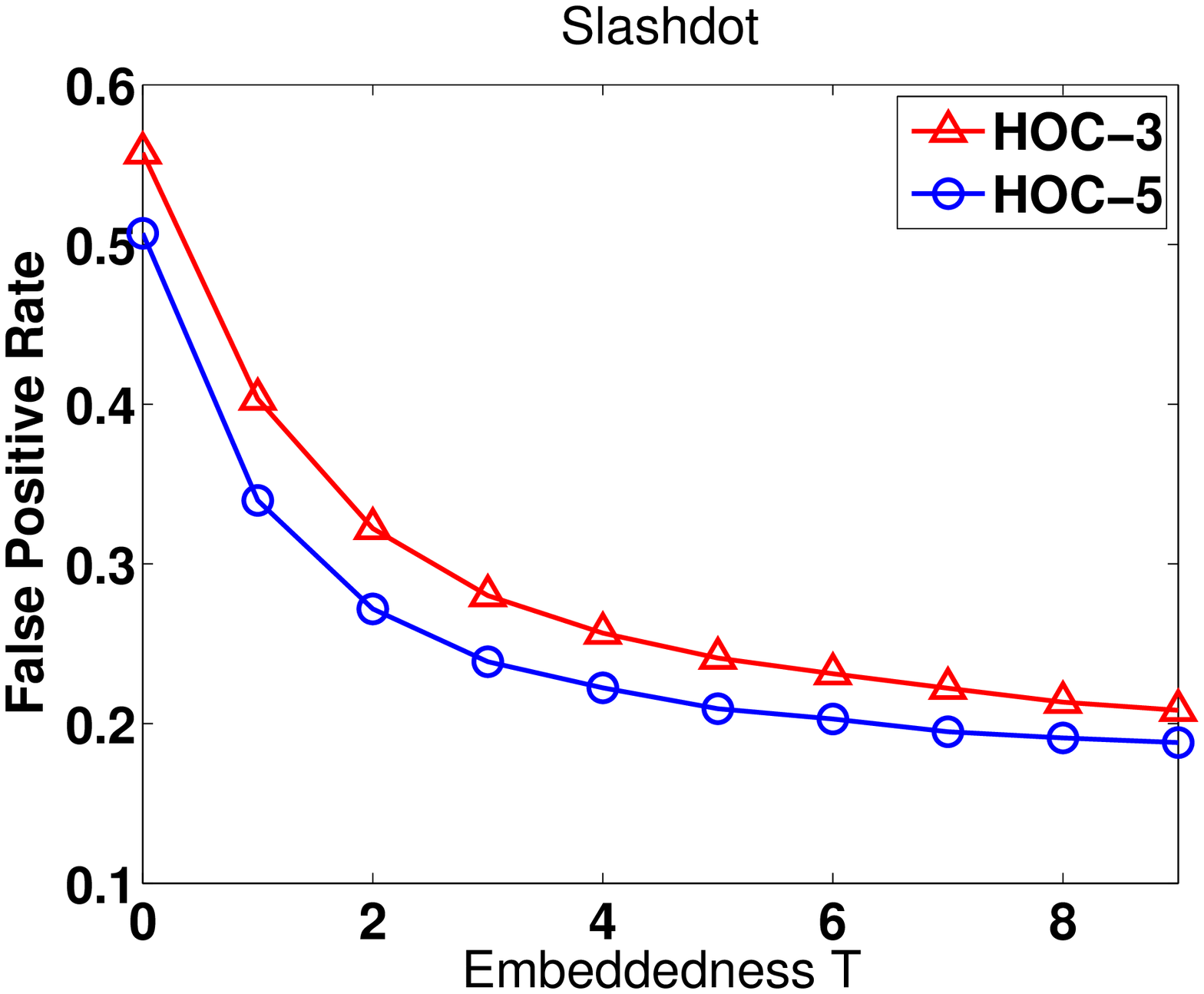}}
\subfloat[Wikipedia]{\includegraphics[width=.3\textwidth]{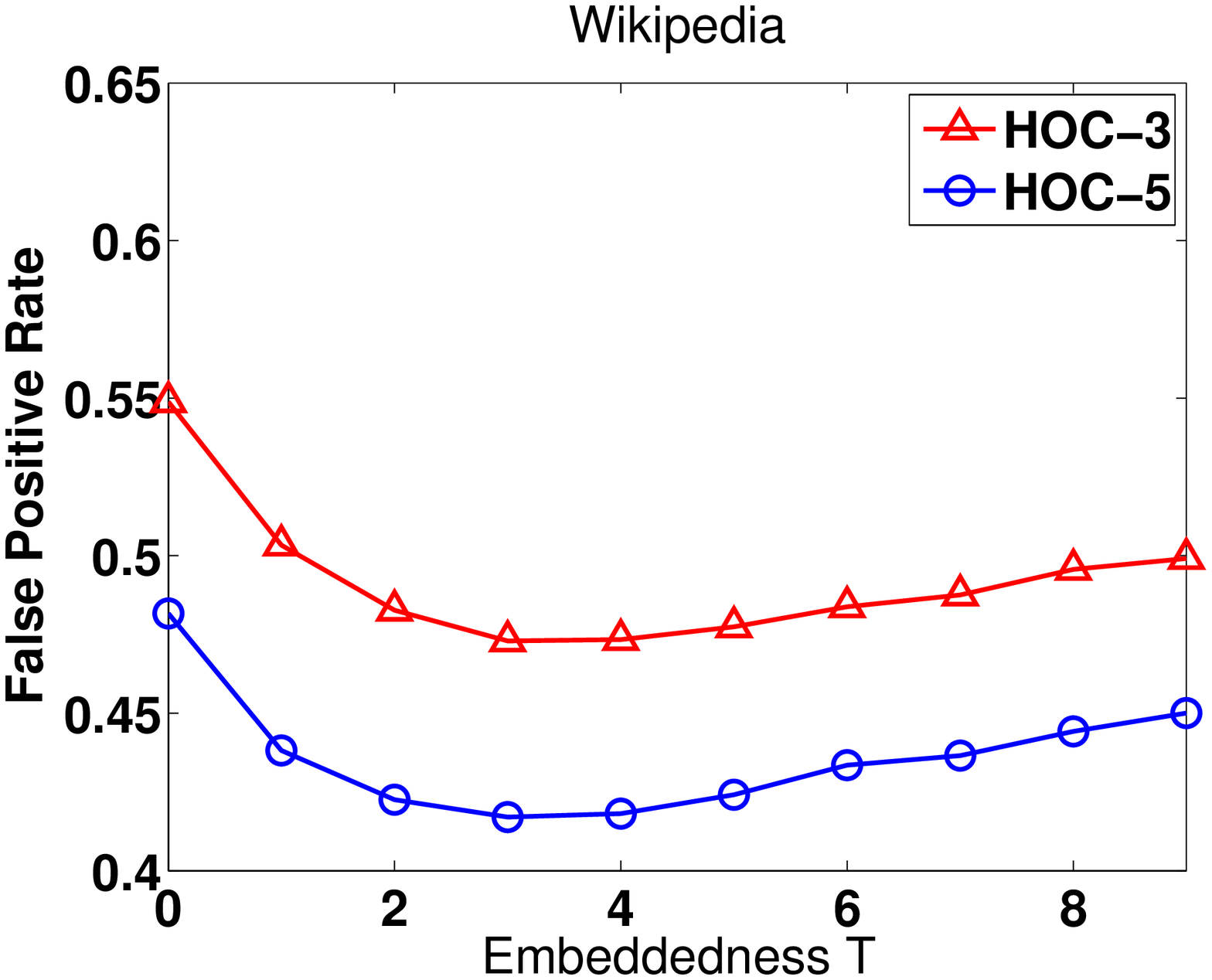}}
\caption{False Positive Rates of Higher Order Cycle (HOC) Methods for $\ell=3, 5.$ \textnormal{These plots show the false positive rate
of HOC-$\ell$ methods for edges with embeddedness {\em at least} $T$ for various thresholds $T$. We see that considering higher order cycles
has the benefit of significantly reducing false-positives while simultaneously achieving slightly better overall accuracy (refer to Table~\ref{tab:prediction_acc}).
However, unlike what we see for MOI methods, here the improvement does not seem to depend strongly on edge embeddedness. The false positive
rates for HOC-$4$ are very similar to that of HOC-$5$ and hence are not shown for clarity.}}
\label{fig:HOC}
\end{figure}

At this point, we see that for cycle-based methods,
considering higher order cycles benefits the accuracy of sign prediction and lowers the false positive rate.
Furthermore, the results are consistent across the three diverse networks. 
These results confirm the intuition that getting more global information improves quality of prediction,
and motivate us to consider the fully global structure of networks.

Now we turn our attention to low rank modeling approaches.
We have seen that LR-SVP fails to perform well under power-law 
distributions of observed relationships in synthetic networks (see Figure \ref{fig:power}),
so we consider the more robust matrix factorization approach for solving the matrix completion problem,
including LR-ALS, LR-SIG and LR-SH, for experiments on real datasets.
Again, we use $10$-fold cross validation setting, and
report the average prediction accuracy for each dataset in Table \ref{tab:prediction_acc}.  
  From the table, 
  we observe that global methods clearly outperform cycle-based methods.  
  In particular, we observe that HOC-5 only improves HOC-3 by less than 1.5\%, while 
  global methods consistently improve the accuracy of HOC-5
  by more than 2\% over all datasets.  In addition,
  LR-SIG and LR-SH further improve the accuracy of LR-ALS.  
  This shows that the sigmoid and square-hinge loss are more suitable
  for sign prediction, which supports the discussion in Section \ref{subsec:MF}.
\begin{table}
  \caption{The sign prediction accuracy for low rank modeling methods
  and cycle-based methods. We can see that the low rank modeling approaches are 
  better than cycle-based methods. }
  \label{tab:prediction_acc}
  \centering
\begin{tabular}{c|cc|cc|ccc}
     & MOI-3 & MOI-10 & HOC-3 & HOC-5 & LR-ALS & LR-SIG & LR-SH \\
     \hline
     Epinions & 0.5539 & 0.8497 & 0.9014 & 0.9080 & 0.9374 & {\bf 0.9465} & 0.9437 \\
     Slashdot & 0.3697 & 0.7850 & 0.8303 & 0.8469 & 0.8774 & 0.8789 & {\bf 0.8835} \\
     Wikipedia & 0.7456 & 0.8220 & 0.8424& 0.8605 & 0.8814 & {\bf 0.8830} & 0.8810 
\end{tabular}
\end{table}

In Figure \ref{fig:all_cmp}, we further select a representative of each category, 
MOI-10, HOC-5 and LR-ALS, and show their performance 
with different levels of edge embeddedness (LR-SIG and LR-SH perform similar to LR-ALS among all datasets).
One might expect that cycle-based approaches should perform better on edges with higher 
embeddedness because more cycle information is available. However, surprisingly
LR-ALS achieves higher prediction accuracy regardless of the embeddedness.
All above results show that global methods are more effective than local methods.

\begin{figure*}[th]
  \centering
  \begin{tabular}{ccc}
    \subfloat[Epinions\label{fig:epi_emd}]{\includegraphics[width=0.3\textwidth,height=3.5cm]{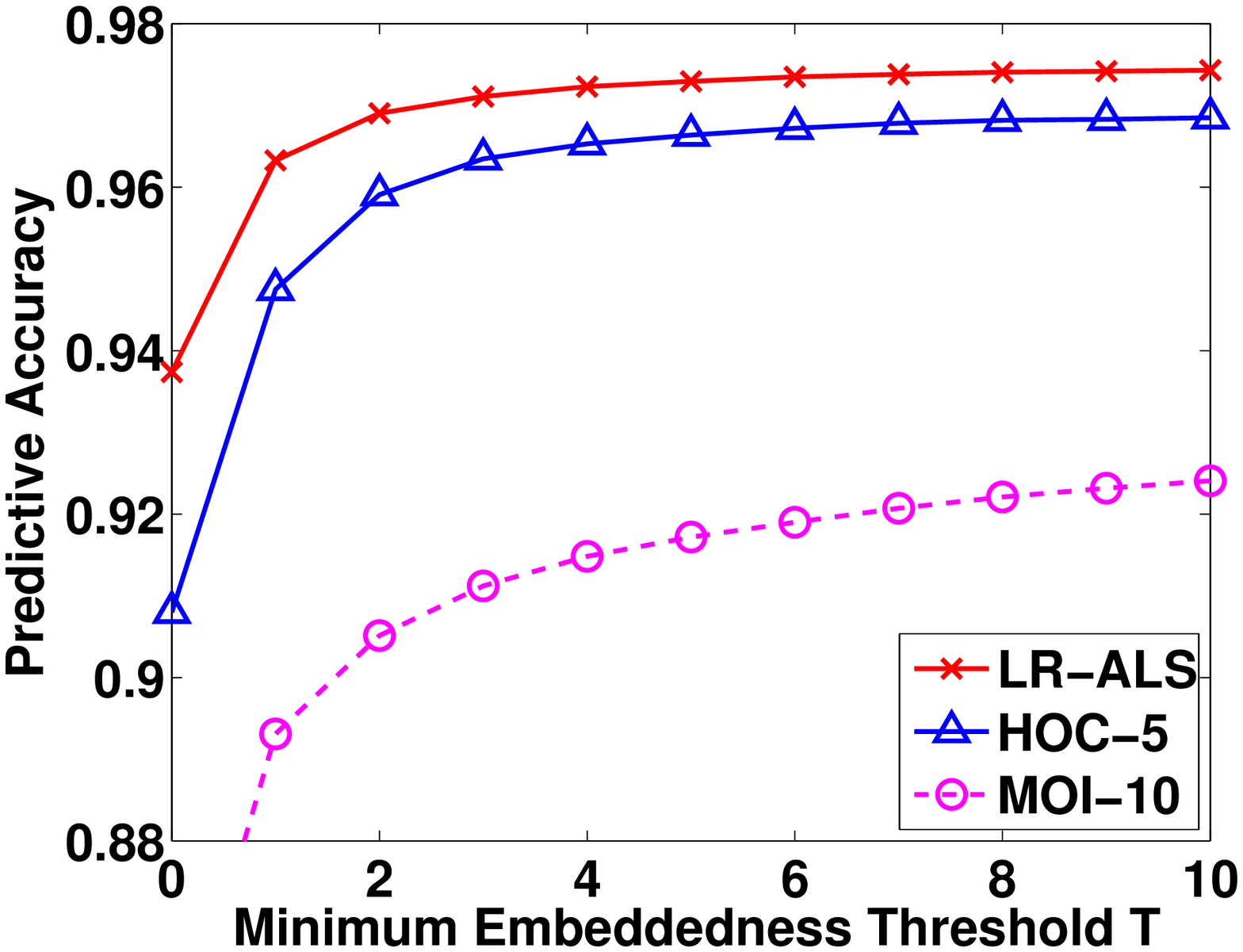}}
 &  
    \subfloat[Slashdot\label{fig:slash_emd}]{\includegraphics[width=0.3\textwidth,height=3.5cm]{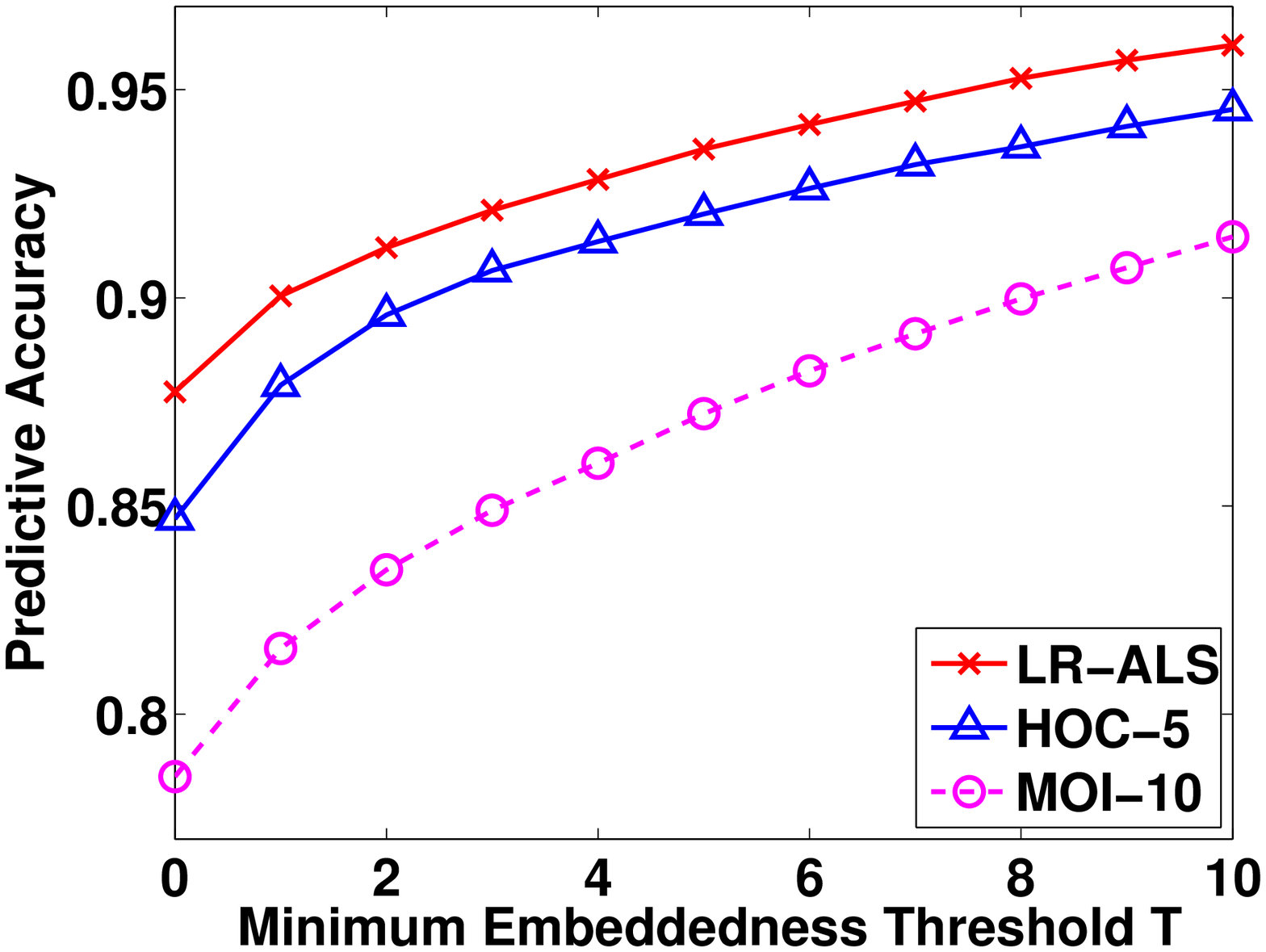}}
    &
    \subfloat[Wikipedia\label{fig:wiki_emd}]{\includegraphics[width=0.3\textwidth,height=3.5cm]{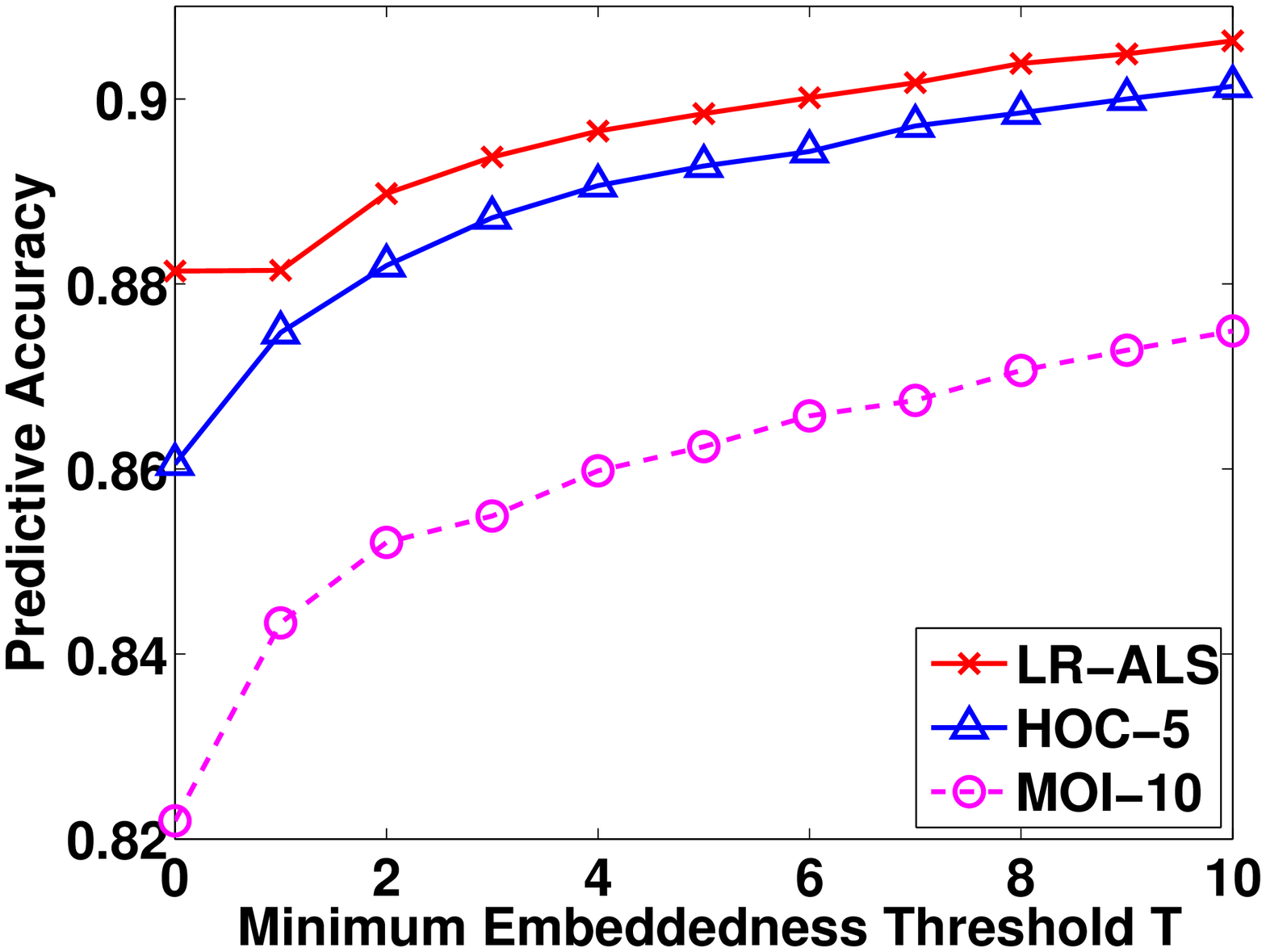}}
  \end{tabular}
  \vspace{-5pt}
  \caption{
  Sign prediction accuracy of local and global methods 
    with different levels of embeddedness.
  These plots show the accuracy for edges with embeddedness at least $T$. 
  We can see that LR-ALS consistently achieves the highest accuracy for all thresholds $T$.}
  \label{fig:all_cmp}
\end{figure*}

\subsubsection{Running Time Comparison}
In addition to prediction accuracy, we now compare the running time required by the 
different methods.  As discussed in Section \ref{subsec:MF}, low rank modeling with
matrix factorization is more efficient than cycle-based algorithms in terms of time complexity. 
Here, we further show that matrix factorization methods are empirically much faster than cycle-based algorithms. 
The running times are summarized in Table \ref{tab:time_compare}.
To conduct timing tests on a large signed network, in addition to the three real 
datasets as described in Table \ref{tab:real_data}, we construct a 
large-scale synthetic dataset called Cluster10 where the number of edges is 100 times more than Epinions. 
Cluster10 is generated from a $10$-weakly balanced network, 
in which clusters have sizes 20000, 40000, \dots, 200000.
There are totally 1.1 million nodes and 120 million edges uniformly sampled
from the complete graph. We construct this synthetic data to show that our matrix 
factorization approach can easily scale up to massive graphs compared to HOC methods.
For matrix factorization approach, we report the time needed to solve the model
by SGD (with sigmoid and square-hinge) and ALS (with square loss).
For HOC methods which build classifiers from cycle-based features,
since the time for training phase depends on 
the classifier, we only report the time for  computation of features.
Thus the reported time for HOC is an underestimate of the time required to construct the HOC model; even then we can see that 
the time required by LR-ALS, LR-SIG and LR-SH is much lower than HOC methods.

In conclusion, for the sign prediction problem, we see that considering
fully global structure of networks gives us the best results.
In particular, the low rank model with matrix factorization is clearly the most
competitive method in terms of accuracy and scalability.

\begin{table}
  \centering
  \caption{Running time (in seconds) for low rank model with matrix factorization 
    and HOC on real datasets and a 1.1 million node synthetic data Cluster10. 
    The time of LR-SGD is the average time of LR-SIG and LR-SH.
    For HOC methods, we only consider the time for feature computation 
    before the model training, while for LR methods we report the total time for
  constructing the model.
  We can see that LR methods with matrix factorization are clearly more 
  efficient than cycle-based algorithms. }
  
  \label{tab:time_compare}
  \begin{tabular}{c|ccccc}
    & HOC-3 & HOC-4 & HOC-5 & LR-ALS & LR-SGD \\
    \hline
Wikipedia & 18.08 & 74.52 & 462.92 & {\bf 2.26} & 2.41\\
Slashdot & 133.4 & 1,936.0 & $> 10,000$ & {\bf 17.4} & 24.7\\
Epinions & 560.64 & 6,156.8 & $> 10,000$ & {\bf 28.67} & 37.2 \\
Cluster10 & $> 10,000$ & $> 10,000$ & $> 10,000$  &  {\bf 455.1} & 1,152 
  \end{tabular}
\end{table}

\subsection{Clustering}
\label{sec:cluster_experiments}
In this subsection, we show that our clustering approach, which completes the low-rank structure
of signed networks before performing clustering, outperforms spectral clustering 
based on the signed Laplacian \citep{Kunegis10a}.
We conduct experiments on synthetic data generated from weakly balanced 
networks (note that we do not have ground truth for clustering in the real-life datasets). We consider a $10$-weakly balanced 
network $\comp$ where size of each group is $100$, and 
observe entries from $\comp$ with two sampling procedures: 
uniform sampling and uniform sampling with noise.


To measure the performance of clustering, 
we calculate the number of edges that satisfy the ground-truth clustering, which 
is defined by 
\begin{equation}
  \sum_{i,j: s_i = s_j} I(\bar{s_i} = \bar{s_j}) + \sum_{i,j: s_i\neq s_j}  
  I(\bar{s_i}\neq \bar{s_j}). 
  \label{eq:clustering_eval}
\end{equation}
where $s_1,\dots,s_n$ denote the ground-truth clustering assignment
for each node, and $\bar{s}_1,\dots,\bar{s}_n$ are the clustering results given by 
the clustering algorithm. 

Following the procedure outlined in the previous subsection,
in the uniform sampling case, we consider the networks $\obs = \comp(s, 0, \uni)$ with $s \in [0.01, 0.06]$,
while in sampling with noise case we consider networks $\obs = \comp(0.1, \epsilon, \uni)$ with $\epsilon \in [0.01, 0.06]$.
For each observed network, we apply Algorithm \ref{alg:clustering} (See Section \ref{sec:clustering}) and 
clustering via the signed Laplacian, and evaluate clustering results using
\eqref{eq:clustering_eval}.
The results of these two scenarios are shown in Figure~\ref{fig:cluster}.
In both the scenarios, our proposed
clustering approach is significantly better than 
clustering based on the signed Laplacian.  This shows that recovering the low-rank structure of signed networks
leads to improved clustering results.

\begin{figure*}[th]
  \centering
  \begin{tabular}{cc}
    \subfloat[\small Data without noise\label{fig:cluster_nonoise}]{\includegraphics[width=0.4\textwidth,height=5cm]{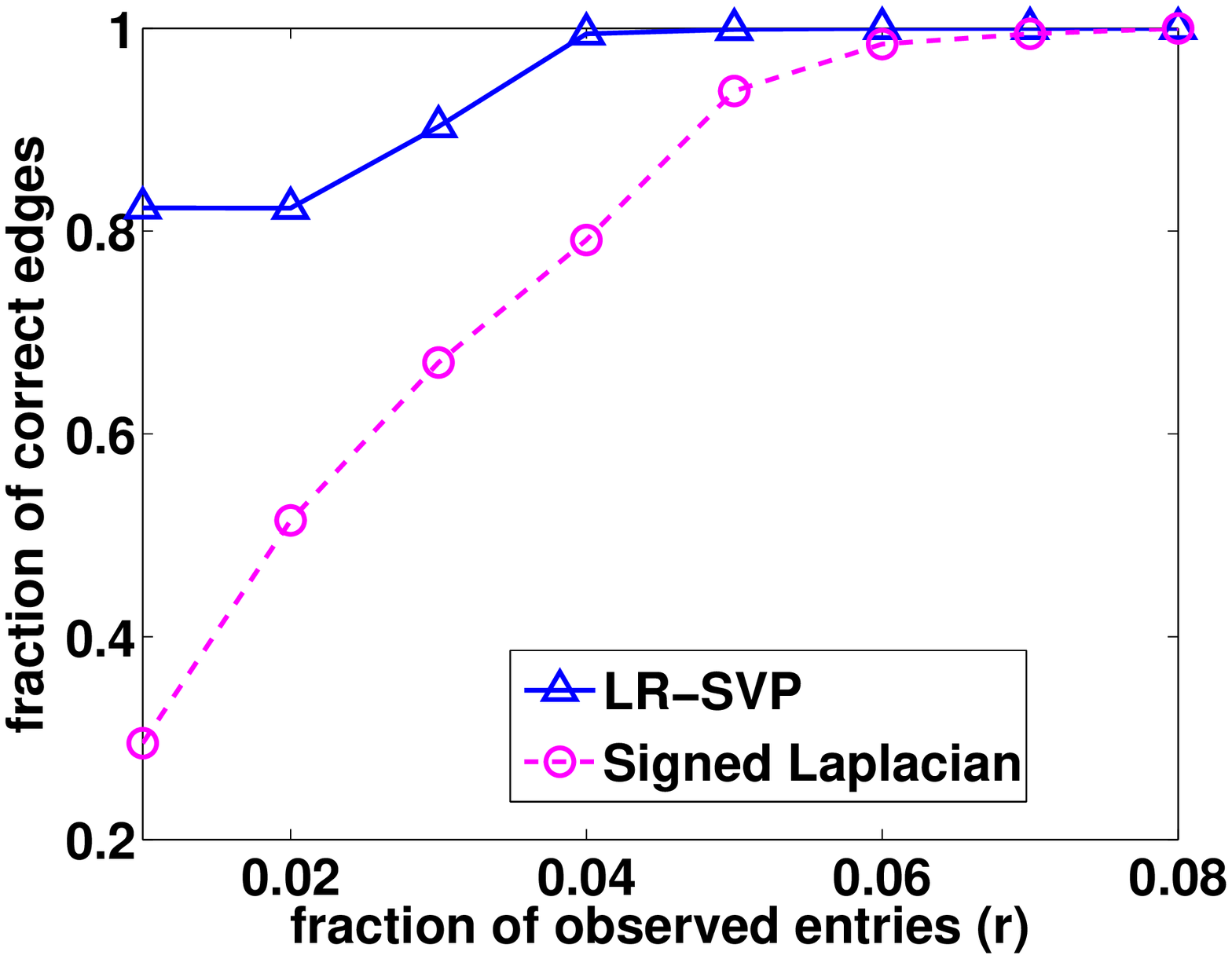}} &
    \subfloat[\small Data with noise\label{fig:cluster_noise}]{\includegraphics[width=0.4\textwidth,height=5cm]{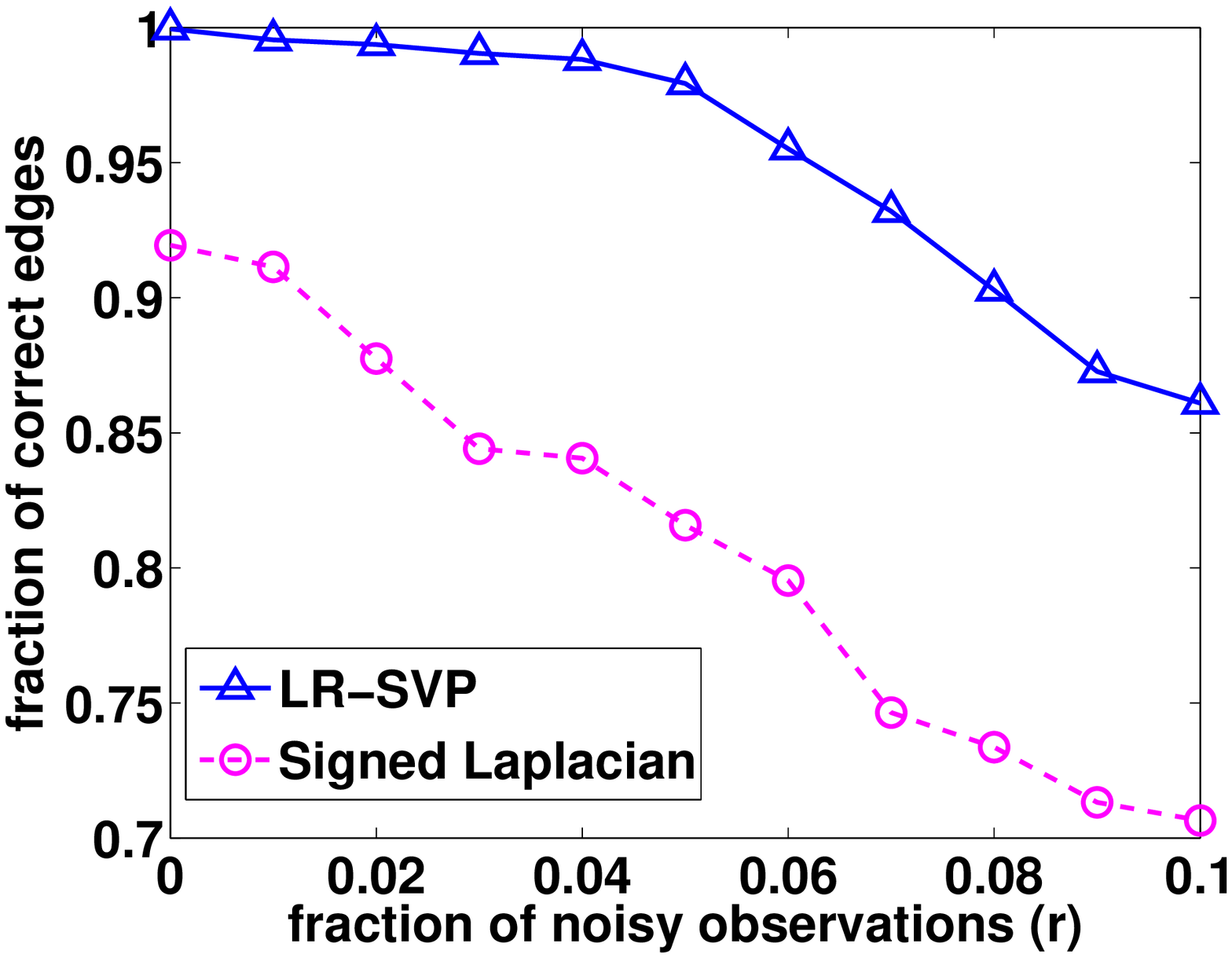}}
  \end{tabular}
  \caption{Clustering partially observed synthetic data. Figure 
  \ref{fig:cluster_nonoise} is the result without noise and Figure 
  \ref{fig:cluster_noise} is the result with noise. In both cases, clustering
  with LR-SVP performs significantly better than clustering with signed Laplacian. }
  \label{fig:cluster}
\end{figure*}



\section{Related Work}
\label{sec:related}
Signed networks have been studied since the early 1950s.  Harary and Cartwright were the first to mathematically study structural balance.  They defined balanced triads and proved the global structure of balanced signed networks \citep{Harary53a, Cartwright56a}.  
\cite{Davis67a} further generalized the balance notion to weak balance by allowing triads with all negative edges, and showed that weakly balanced graphs have mutual antagonistic groups as global structure.

Though theoretical studies of signed networks have been conducted for a long time,
it was not until this decade that analysis of real signed networks could be
done at a large scale as large real networks have become more accessible recently.
For example, \cite{Kunegis09a} performed several analysis tasks on Slashdot, and
\cite{Leskovec10a, Leskovec10b} studied the local and global structure of three real
signed networks.  They designed several computational experiments to justify
that the structure of these signed networks match some widely believed social theories.

In this paper, we focused on problems in signed networks.
However, these problems have their counterparts in unsigned networks.
For instance, structural link prediction in unsigned networks 
corresponds to the sign prediction problem.  Structural link prediction has been well explored, 
and it is usually solved by computing a similarity measure between nodes 
\citep{Nowell07a}, such as those proposed by \cite{Katz53a} and \cite{Adamic03a}.
The sign prediction problem, however, was not formally considered until 
the work by \cite{Guha04a}, in which they develop a trust propagation framework to 
predict trust or distrust between entities.  
More recently,  \cite{Kunegis09a, Kunegis10a} reconsidered this problem
by using various similarity functions and kernels such as 
matrix exponential and signed Laplacian. 
\cite{Leskovec10a} proposed a machine learning formulation of this problem, 
arguing that learning from only local triangular structure of edges 
can achieve reasonable accuracy.  

Sign prediction using our global method is closely related to the low-rank matrix completion problem.  
In the last five years,
there has been substantial research studying exact recovery conditions 
for this problem \citep{Candes08a, Candes09a}, and algorithms with theoretical 
guarantees have also been proposed, such as SVT \citep{Cai10a} and SVP \citep{Jain10a}.  
Matrix factorization is another approximation technique for matrix 
completion.  Though this approach is notoriously hard to analyze, it
is very competitive in practice \citep{Koren09a}.
While the matrix completion problem has been considered mostly
in collaborative filtering, our low rank model arises naturally from the weak balance of signed networks.  

Clustering is another fundamental problem in network analysis.  For unsigned
networks, there are several proposed algorithms that have been shown to be effective, 
such as clustering via graph Laplacian \citep{Ng01a}, modularity \citep{Newman06a} 
and multilevel approaches \citep{Dhillon07a}.
However, most of these approaches can not be directly extended to signed networks
since weak balance theory does not apply to unsigned networks.  
As a result, researchers have tried to tailor unsigned network clustering algorithms 
in order to make them applicable to signed networks.  
For instance, \cite{Doreian96a} proposed a
local search strategy which is similar to the Kernighan-Lin algorithm 
\citep{Kernighan70a}.
Starting with an initial clustering assignment, 
it tries to move nodes one by one to get a more preferable clustering.  
\cite{Yang07a} proposed an agent-based method which basically conducts 
a random walk on the graph.  
\cite{Kunegis10a} generalized spectral algorithms to signed networks.
They proposed a spectral approach using the so-called ``signed'' Laplacian, 
and showed that partitioning signed networks into two groups using the signed Laplacian kernel 
is analogous to considering ratio cut on unsigned networks. 
\cite{Anchuri12a} proposed hierarchical 
iterative methods that solve 2-way signed modularity objectives 
using spectral relaxation at each hierarchy.
\cite{Chiang12a} proposed some graph kernels for signed network clustering, 
and showed that the multilevel framework can be extended to this problem.

Another line of research on signed graph clustering problem is correlation clustering.
Correlation clustering is motivated from document classification: given a set of documents
with some pairs of documents labeled similar or different, the goal is 
to find a partition such that documents in the same cluster are mostly similar \citep{Bansal04a}.
The problem was first considered by \cite{Bansal04a}, who proved that the problem is NP-hard to optimize, and proposed two approximation algorithms to maximize the ``agreement" (defined as the number of edges that are correctly classified under a partition) and minimize the ``disagreement" (defined vice versa) under the special case that all pairwise label information is given.  
The bounds for general correlation clustering setting were provided by \cite{Demaine05a}.  
On the other hand, some researchers have also considered the correlation clustering problem from the 
statistical learning theory viewpoint.  For example, \cite{Thorsten05a} give error bounds for the problem if only partial pairs are observed.
Recently, \cite{Cesa12a} proposed a method for sign prediction by learning a
correlation clustering index.  They consider three types of learning models: batch, online
and active learning, and provide theoretical bounds for prediction mistakes under each setting.


\section{Conclusions and Future Work}
\label{sec:conclusion}
In this paper, we studied the usefulness of social balance on signed networks, 
with two fundamental applications: sign prediction and clustering.
Starting from a local view of social balance, we proposed two families
of sign prediction methods based on local triads and cycles:
prediction via measures of social imbalance (MOIs) and supervised learning
based on high order cycles (HOCs).  For both approaches,  
predictive accuracies are improved if longer cycles are taken into
consideration, suggesting that a broader view of local patterns helps
in sign prediction.  We then considered the fully global perspective on social 
balance, and showed that the adjacency matrices of balanced networks are
low rank.  Based on this observation, we modeled the sign prediction problem as a low-rank
matrix completion problem.  We discussed
three approaches to matrix completion: convex
relaxation, singular value projection, and matrix factorization.
In addition, we
applied this low rank modeling technique to the clustering problem.
In experiments, we observe that sign prediction via matrix factorization
not only outperforms MOIs and HOCs, but requires much less running time.
Clustering results are also more favorable via the matrix
completion approach in comparison with the existing signed Laplacian approach.
All of these results consistently demonstrate the effectiveness of the global
viewpoint of social balance.

For future work, one possible direction is to explore analysis tasks
on heterogeneous signed networks.  Since there are different types of
entities in heterogeneous networks, currently there are no clear answers to questions such as: do balance relationships
exist on such networks? How do we quantitatively measure balance 
if balance patterns exist? How is balance at a local level related to the global structure?
Furthermore, another possible direction is to examine other theories 
for analysis tasks on signed networks.  For example, some recent work \citep{Leskovec10a, Leskovec10b} has considered status theory.  While
\cite{Leskovec10a} found evidence that status theory holds in general
in real signed networks, patterns conforming to status theory are quite different
from those conforming to balance theory.  Thus, it is natural to ask 
how to design algorithms by pursuing global patterns conforming to status theory.
These interesting directions are worth exploring in future research.


\section{Acknowledgments}
We gratefully acknowledge the support of NSF under grants CCF-0916309, CCF-1117055, and of DOD Army under grant W911NF-10-1-0529. Most of the contribution of Ambuj Tewari to this work occurred while he was a postdoctoral fellow at the University of Texas at Austin.




\bibliography{sign_arxiv}
\bibliographystyle{plainnat}
\end{document}